\documentclass[12pt]{article}

\usepackage[bookmarksopen=true,bookmarks=true,unicode,setpagesize]{hyperref}
\hypersetup{colorlinks=true,linkcolor=black,citecolor=black}
\usepackage{amsmath,amsthm,amssymb}
\allowdisplaybreaks

\newtheorem{thm}{Theorem}
\newtheorem{prop}[thm]{Proposition}
\newtheorem{lemma}[thm]{Lemma}
\newtheorem{cor}[thm]{Corollary}
\theoremstyle{remark}
\newtheorem{rem}[thm]{Remark}
\theoremstyle{remark}

\theoremstyle{remark}

\newcommand{\Ker}{\operatorname{Ker}}

\newcommand{\td}{T^{(2)}}
\newcommand{\tn}{T^{(n)}}

\newcommand{\fn}{f^{(n)}}

\newcommand{\Fnhc}{{\mathcal{F}}_n{({\h})}}
\newcommand{\Fhc}{{\mathcal{F}}{({\h})}}
\newcommand{\Fhcfin}{{\mathcal{F}_{\mathrm{fin}}}{({\h})}}

\newcommand{\h}{\mathcal{H}}

\newcommand{\hcm}{\mathcal{H}^{\otimes n}}

\newcommand{\N}{\mathbb{N}}

\newcommand{\pn}{\mathcal{P}_n}
\newcommand{\prd}{\partial^{\dagger}}

\newcommand{\pr}{\partial}

\newcommand{\Ran}{\operatorname{Ran}}

\newcommand{\vp}{\varphi}

\begin{document}

\begin{center}{\Large \bf
  Fock representations  of $Q$-deformed commutation relations
}\end{center}

{\large Marek Bo\.zejko}\\
Institute of Mathematics Polish Academy of Science. Kopernika 18, 51-617 Wroclaw, Poland\\
e-mail: \texttt{bozejko@math.uni.wroc.pl}\vspace{2mm}

{\large Eugene Lytvynov}\\ Department of Mathematics,
Swansea University, Singleton Park, Swansea SA2 8PP, U.K.\\
e-mail: \texttt{e.lytvynov@swansea.ac.uk}\vspace{2mm}

{\large Janusz Wysocza\'nski}\\
Instytut Matematyczny, Uniwersytet Wroc{\l}awski, Pl.\ Grunwaldzki 2/4, 50-384 Wroc{\l}aw, Poland\\
e-mail: \texttt{jwys@math.uni.wroc.pl}\vspace{2mm}

%{\bf Corresponding author:} Janusz Wysocza{\'n}ski,  \texttt{janusz.wysoczanski@math.uni.wroc.pl}\\
%Tel.: +48 71 3757095

\vspace{2mm}

{\small

\begin{center}
{\bf Abstract}
\end{center}
\noindent
We consider Fock representations of the $Q$-deformed commutation relations
$$\partial_s\partial^\dag_t=Q(s,t)\partial_t^\dag\partial_s+\delta(s,t), \quad s,t\in T.$$
Here
 $T:=\mathbb R^d$   (or more generally $T$ is a locally compact Polish space), the function $Q:T^2\to \mathbb C$ satisfies $|Q(s,t)|\le1$ and $Q(s,t)=\overline{Q(t,s)}$, and
$$\int_{T^2}h(s)g(t)\delta(s,t)\,\sigma(ds)\sigma(dt):=\int_T h(t)g(t)\,\sigma(dt),$$
  $\sigma$ being a fixed reference measure on $T$. In the case where $|Q(s,t)|\equiv 1$, the  $Q$-deformed commutation relations describe a generalized statistics studied by Liguori and Mintchev (1995). These generalized statistics contain anyon statistics as a special case (with $T=\mathbb R^2$ and a special choice of the function $Q$).
The related $Q$-deformed Fock space $\mathcal F(\mathcal H)$ over $\mathcal H:=L^2(T\to\mathbb C,\sigma)$  is constructed. An explicit form of the  orthogonal projection of $\mathcal H^{\otimes n}$ onto the $n$-particle space $\mathcal F_n(\mathcal H)$ is derived. A scalar product in $\mathcal F_n(\mathcal H)$ is given by an operator $\mathcal P_n\ge0$ in $\mathcal H^{\otimes n}$ which is strictly positive on $\mathcal F_n(\mathcal H)$. We realize the smeared operators $\partial_t^\dag$ and  $\partial_t$ as creation and annihilation operators in $\mathcal F(\mathcal H)$, respectively.
 Additional $Q$-commutation relations are obtained between the creation operators   and between the annihilation operators. They are of the form $\partial^\dag_s\partial^\dag_t=Q(t,s)\partial^\dag_t\partial^\dag_s$, $\partial_s\partial_t=Q(t,s)\partial_t\partial_s$, valid for those $s,t\in T$ for which $|Q(s,t)|=1$.
 }
\vspace{2mm}


\section{Introduction}
The aim of the paper is to construct Fock representations of the $Q$-commutation relations
\begin{equation}\label{QCR_0}
\pr_s\prd_t=Q(s,t)\prd_t\pr_s+\delta(s,t),\quad s,t\in T.
\end{equation}
Here $T=\mathbb R^d$, or more generally, $T$ is a locally compact Polish space, the function $Q:T^2\to \mathbb C $ is Hermitian, i.e., $Q(s,t)=\overline{Q(t,s)}$, and satisfies $|Q(s,t)|\le1$,  $\partial_t$ and  $\partial_t^\dag$ are operator-valued distributions, adjoint of each other, and
$$ \int_{T^2} \delta(s,t) f(s,t)\,\sigma(ds)\,\sigma(dt):=\int_T f(t,t)\,\sigma(dt),$$
where $\sigma$ is a fixed  Radon meaure on $X$ (typically $\sigma(dt)=dt$ being the Lebesgue measure if $T=\mathbb R^d$). We will call \eqref{QCR_0} the {\it $Q$-deformed commutation relations}, or just $Q$-CR.

For a function $Q$ satisfying $|Q(s,t)|\equiv 1$, a Fock representation of the $Q$-CR was constructed by Liguori and Mintchev \cite{liguoriminchevCMP}. In that  case, creation operators $\partial_t^\dag$ and annihilation operators $\partial_t$ satisfy the additional commutation relations:
\begin{equation}\label{gyut86t5}
\prd_s\prd_t=Q(t,s)\prd_t\prd_s,\quad \pr_s\pr_t=Q(t,s)\pr_t\pr_s.\end{equation}
The term {\it Fock representation}  means that, for each annihilation operator, one has $\pr_t\Omega=0$, where $\Omega$ is the vacuum vector.

In the present study, relations \eqref{gyut86t5} will hold for those $s,t\in T$ which  satisfy $|Q(s,t)|=1$. Note that, under the assumption that the function $Q$ is Hermitian, the commutation relations \eqref{gyut86t5} are consistent if and only if $|Q(s,t)|=1$.

For the first time, an interpolation between the canonical (bosonic) commutation relations (CCR) and the canonical (fermionic) anticommutation relations (CAR) was rigorously constructed in  \cite{bozejkospeicherCMP1991}.
Let $\mathcal H$ be a separable Hilbert space and let $q\in(-1,1)$.  On a $q$-deformed Fock space over $\mathcal H$, Bo\.zejko and Speicher \cite{bozejkospeicherCMP1991} constructed $q$-creation operators $a^{+}(f)$ (in fact $a^{+}(f)$ were \textit{free} creation operators), and  $q$-annihilation operators $a^{-}(f):=(a^+(f))^*$,  for $f\in \h$, which satisfy the $q$-commutation relations:% between the creation and the annihilation operators:
\begin{equation}\label{iyt867}
a^-(f)a^+(g)=qa^+(g)a^-(f) + (f, g)_{\h}, \quad f,g \in \h.
\end{equation}
The limits $q= 1$ and $q=-1$ correspond to the boson and fermion statistics, respectively, thus giving the CCR and CAR.
The case $q=0$ corresponds to the creation and annihilation operators acting in the full Fock space; these operators are particularly important for models of free probability, see e.g.\ \cite{NiSp,bozejko-lytvynovCMP2009}. Aspects of noncommutative probability related to the general $q$-commutation relations \eqref{iyt867} were discussed e.g.\ in \cite{bozejkospeicherCMP1991,bozejkokummererspeicherCMP1997,anshelevichDocMath2001}.

By using probabilistic methods,  Speicher \cite{speicherLMPh1993} proved existence of a representation of  the (discrete) $q_{ij}$-commut\-ation relations  of the form
\begin{equation}\label{ghdtrdeyk}
\partial_i\partial^\dag_j=q_{ij}\partial^\dag_j\partial_i+\delta_{ij}
\end{equation}
with $-1 \leq q_{ij}=q_{ji}\leq 1$,  $i,j \in \N$, and $(\partial^\dag_i)^*=\partial_i$.
 Bo\.zejko and  Speicher  \cite{bozejkospeicherMathAnn1994}
 constructed a Fock representation of the following commutation relations between creation operators $\partial^\dag_j$ and annihilation operators $\partial_i$, with $i, j\in \N$:
\begin{equation}\label{gtf67er}
\partial_i\partial_j^\dag=\sum_{k,l}q^{ik}_{jl}\,\partial^\dag_k\partial_l+ \delta_{i,j}.
\end{equation}
 They showed that, if the operator
$\Psi$ given by the matrix $(q^{ik}_{jl})_{i,j,k,l}$ is
self-adjoint, satisfies the braid relations, and has norm $<1$,  then there exists a Fock  representation of the  commutation relations \eqref{gtf67er}. As a consequence, they obtained a Fock  representation of the $q_{ij}$-commutation relations \eqref{ghdtrdeyk} even for complex $q_{ij}$ with $\overline{q_{ij}}=q_{ji}$ and  $\sup_{i,j}|q_{ij}|=\|\Psi\|<1$. By taking the weak limit of corresponding operator algebras,  Bo\.zejko and  Speicher  \cite{bozejkospeicherMathAnn1994} also derived existence
of a representation of the $q_{ij}$-commutation relations \eqref{ghdtrdeyk} with $\sup_{i,j}|q_{ij}|=\|\Psi\|=1$.
Also J{\o}rgensen, Schmitt and Werner  \cite{jorgensenschmittwernerPacific1994,jorgensenschmittwernerJFA1995} considered representations of the   commutation relations~\eqref{gtf67er}.

In the case where $\|\Psi\|=1$,  J{\o}rgensen, Proskurin, and Samo\v{\i}lenko \cite{jorgensenproskurinsamoilenkoPacificJM} found, for $n\ge2$, the kernel of the nonnegative operator  which determines the scalar product in the $n$-particle space of the Fock space corresponding to the commutation relations \eqref{gtf67er}. The papers \cite{bozejkospeicherMathAnn1994} and  \cite{jorgensenproskurinsamoilenkoPacificJM} taken together give then a Fock representation of the commutation relations \eqref{ghdtrdeyk} with $\sup_{i,j}|q_{ij}|=1$.

Properties of the algebras generated by such operators were studied by many authors. In the context of $C^*$-algebras, let us  mention  the works by Dykema and Nica \cite{dykemanicaJReineAngew1993} and Kennedy and Nica \cite{kennedynicaCMP2011} (who studied relations of the $C^*$-algebras generated by the $q$-commutation relations with the Cuntz algebra), J{\o}rgensen, Schmitt and Werner  \cite{jorgensenschmittwernerPacific1994,jorgensenschmittwernerJFA1995} (who studied the Wick order generated $C^*$-algebras), Proskurin and Samoilenko \cite{proskurinsamoilenkoJMS2010} (who studied general Wick *-algebras).  There are also a number of studies of the $q$-commutation relations in the context of von Neumann algebras, in particular, by Lust-Piquard \cite{lustpiquard1999} (who studied properties of the Riesz transform), Kr\'olak \cite{krolakCMP2000}, Nou \cite{nouMathAnn2004}, {\'S}niady \cite{sniadyCMP2004}, Ricard \cite{ricardCMP2004} (who studied factoriality problems), Shlyakhtenko \cite{shlakhtenkoIMRN2004} (who studied Voiculescu's free entropy for families of $q$-Gaussian operators), and Bo\.zejko \cite{bozejkoDemMath2012} (who studied positivity of the symmetrization operators constructed through a self-adjoint Yang--Baxter operator $\Psi\geq -1$). Also Dabrowski \cite{Da2014}, Guionnet and Shlyakhtenko \cite{guionnetshlyakhtenkoInventMath2014}, and  Nelson and Zeng \cite{NZ2015a,NZ2015b} proved that $q$-factors or, more generally, $q_{ij}$-factors are isomorphic to the free group factors ($q=0$) for small values of $q$ or $q_{ij}$, respectively. Another possible generalization of the commutation relations \eqref{iyt867} related to the group of signed permutations can be found in  \cite{BEH2015}

All the above mentioned investigations are of \textit{discrete} type, so that the set $T$ is at most countable.
As we have already mentioned above, in the \textit{continuous} setting, a Fock representation of the $Q$-CR \eqref{QCR_0}, \eqref{gyut86t5}, called a {\it generalized statistics}, was constructed by Liguori and Mintchev \cite{liguoriminchevCMP}, see also \cite{GMS1980,GMS1981,goldinsharpPhRevLett1996,goldinmajidJMP2004,bozejkolytvynovwysoczanskiCMP2012}. A rigorous meaning of these commutation relations is given by smearing them with functions from $\mathcal H:=L^2(T\to\mathbb C,\sigma)$. More precisely,
defining for $f\in\mathcal H$ operators $a^+(f):=\int_T \sigma(dt)\, f(t)\partial^\dag_t$ and $a^-(f):=\int_T \sigma(dt)\, \overline{f(t)}\partial_t $, we get the
commutation relations:
\begin{align*}
a^-(f)a^+(g)&=\int_{T^2} \sigma(ds)\,\sigma(dt) \,\overline{f(s)}g(t)Q(s,t)\partial_t^\dag\partial_s+\int_T \overline{f(t)}g(t)\,\sigma(dt),
\\
 a^+(f)a^+(g)&=\int_{T^2} \sigma(ds)\,\sigma(dt) \,f(s)g(t)Q(t,s)\partial_t^\dag\partial_s^\dag, \\
a^-(f)a^-(g)&=\int_{T^2} \sigma(ds)\,\sigma(dt) \,\overline{f(s)g(t)}Q(t,s)\partial_t\partial_s,
\end{align*}
where $f,g\in\mathcal H$. (Of course, the operator-valued integrals in these relations should be given a rigorous meaning.)

From the physical point of view, the most important case  of a generalized statistics is the anyon statistics, where $T=\mathbb R^2$ and the function $Q(s,t)$ is determined by a complex parameter $q$ with $|q|=1$, namely,
\begin{equation}\label{ytre6ie5i9}
Q(s,t) = \begin{cases}
q,&\text{if }s^1<t^1,\\
\bar q, &\text{if }s^1>t^1.
\end{cases}
\end{equation}
Here, $s=(s^1,s^2), t=(t^1,t^2)\in\mathbb R^2$. Note that the value of the function $Q$ on the set $\{(s,t)\in T^2\mid s^1=t^1\}$ does not matter for the Fock representation of the $Q$-CR.
For an explanation as to why such commutation relations describe an anyon statistic, we refer the reader to Liguori and Mintchev's paper \cite{liguoriminchevCMP}
and to Goldin and Sharp's paper \cite{goldinsharpPhRevLett1996}.

Goldin and Majid \cite{goldinmajidJMP2004} proved the following anyonic exclusion principle, which generalizes Pauli's exclusion principle for fermions: If $q^m=1$ and $q\ne1$, then the creation operators $a^+(f)$ in the Fock representation of the anyon commutation relations satisfy $a^+(f)^m=0$, or equivalently, the $Q$-symmetrization of the function $f^{\otimes m}$ is equal to zero.

In \cite{Ly2016}, non-Fock representations of the anyon commutation relations have been constructed, whose vacuum states are gauge-invariant quasi-free. Note that, for those representations, the (real) value of the function $Q(s,t)$ for $s=t$ must be specified.

Let us  mention that anyon systems have also been considered in the discrete setting, i.e., when $T\subset\mathbb N$, see e.g.\  \cite{FSSS,LS1993,goldinmajidJMP2004}. It should be, however, mentioned that, when discussing the anyons in the discete setting, Goldin and Majid \cite{goldinmajidJMP2004}
dropped the assumption that the annihilation operator is adjoint  of the creation operator, and proved an anyonic exclusion principle for their model.

In this paper, we study
 the continuous case with a function $Q$ satisfying $|Q(s,t)|\le 1$. This natural choice of  $Q$ contains generalized statistics and the relations \eqref{iyt867} as special cases. We would also like to draw the reader's attention to the study by Merberg \cite{merberg2012}, where the case $Q: T^2\to (-1,1)$ was considered and factoriality of the related von Neumann algebras generated by the $Q$-Gaussian operators was discussed.

In Section 2 we present a construction of the Fock representation of the $Q$-CR \eqref{QCR_0}. To this end, we construct a certain
 $Q$-deformed Fock space over $\mathcal H=L^2(T\to\mathbb C,\sigma)$, denoted by $\mathcal F(\mathcal H)$. We  describe the $n$-particle subspaces, $\mathcal F_n(\mathcal H)$, of  $\mathcal F(\mathcal H)$. As a set, each $\mathcal F_n(\mathcal H)$ is a subset of $\mathcal H^{\otimes n}=L^2(T^n\to\mathbb C,\sigma^{\otimes n})$ and consists of all functions $f^{(n)}\in \mathcal H^{\otimes n}$ that are {\it $Q$-quasisymmetric}, meaning that,   a.e.\ for each $k\in\{1,\dots,n-1\}$,
  \begin{equation}\label{Q-symmetric functions}
f^{(n)}(t_1, \dots , t_n) = Q(t_k,
t_{k+1})f(t_1, \dots , t_{k+1}, t_k,  \dots , t_n)
\end{equation}
 provided $|Q(t_k,t_{k+1})|=1$.  We derive an explicit formula for the orthogonal projection of $\mathcal H^{\otimes n}$ onto $\mathcal F_n(\mathcal H)$. A scalar product in $\mathcal F_n(\mathcal H)$ is given by an operator $\mathcal P_n\ge0$ in $\mathcal H^{\otimes n}$ which is strictly positive on $\mathcal F_n(\mathcal H)$.
   %An important result that we use in our considerations is a theorem by  J{\o}rgensen, Proskurin, and Samo\v{\i}lenko \cite{jorgensenproskurinsamoilenkoPacificJM} which allows us to find  the kernel of the operator $\mathcal P_n$ in   $\mathcal H^{\otimes n}$.
   We then realize $a^+(f)$, $a^-(f)$ ($f\in\mathcal H$) as creation and annihilation  operators acting in the $Q$-deformed Fock space $\mathcal F(\mathcal H)$. These operators satisfy
 the $Q$-CR \eqref{QCR_0}. Additionally, due to the $Q$-symmetry \eqref{Q-symmetric functions} in each $\mathcal F_n(\mathcal H)$, we get the following commutation relations between the creation operators and between the annihilation operators:
\begin{align}
\prd_s\prd_t&=Q(t,s)\prd_t\prd_s,\quad\text{if }|Q(s,t)|=1,\notag\\
\pr_s\pr_t&=Q(t,s)\pr_t\pr_s,\quad \text{if }|Q(s,t)|=1.\label{vyi76}
\end{align}

We  note that, by choosing $T$ to be a discrete set and $\sigma$ to be the counting measure on $T$, one can apply our results in a discrete setting. In fact,  the explicit description of the $n$-particle space $\mathcal F_n(\mathcal H)$, explicit formula for the orthogonal projection of    $\mathcal H^{\otimes n}$ onto $\mathcal F_n(\mathcal H)$, and the additional commutation relations \eqref{vyi76} appear to be new results even in the discrete setting.

We finish  Section~2 with a proposition that shows that
discrete anyons of fermion type satisfy  the anyonic exclusion principle, compare with \cite{goldinmajidJMP2004}.

In Section 3, we prove the results formulated in Section~2.

\section{Construction of the Fock representation of\\ $Q$-CR}
In this section, we will construct a Fock representation of the commutation relation \eqref{QCR_0}, and we will note that  the additional commutation relations \eqref{vyi76} then also hold.

\subsection{Operator $\mathcal P_n$} \label{ydr6ur5ei}

Let $T$ be a locally compact Polish space, let $\mathcal B(T)$ denote the Borel $\sigma$-algebra on $T$, and let $\sigma$ be  a Radon measure on $(T,\mathcal B(T))$. Let $E\in\mathcal B(T^2)$ be a symmetric subset of $T^2$:
 if $(s, t)\in E$ then $(t, s)\in E$. We assume  that $\sigma^{\otimes 2}(E)=0$. Denote $T^{(2)}:=T^2\setminus E$, which is also a symmetric set.
 We fix a
complex-valued measurable function
\[
Q:\td \to \{z\in \mathbb{C}: |z|\leq 1 \}
\] which is Hermitian: for all $(s, t)\in \td$, we have $Q(s, t) = \overline{Q(t, s)}$. This function is defined $\sigma^{\otimes 2}$-almost everywhere on
$T^2$.

\begin{rem}
 The case where $|Q(s,t)|=1$ for all $(s,t)\in T^{(2)}$  corresponds to a generalized statistics studied by Liguori and Mintchev \cite{liguoriminchevCMP}. The special case where $T=\mathbb R^2$, $\sigma(dt)=dt$ is the Lebesgue measure on $T$, $E=\{(s,t)\in T^2\mid s^1=t^1\}$, and  the function $Q$ is defined by formula \eqref{ytre6ie5i9} with $q\in\mathbb C$, $|q|=1$,  corresponds to anyon statistics, see  \cite{liguoriminchevCMP, goldinsharpPhRevLett1996,goldinmajidJMP2004}.
The choice $Q(s,t)=q$  for all $(s,t)\in T^{(2)}=T^2$ with $q\in(-1,1)$  corresponds to the $q$-commutations \eqref{iyt867}, see \cite{bozejkospeicherCMP1991}. \end{rem}

%\begin{example} Choose $A=\varnothing$, so that $T^{(2)}=T^2$ and set $Q(s,t)\equiv q$, where $q\in(-1,1)$. This choice of the function $Q$ corresponds to the $q$-commutation relations. The Fock representation of these commutation relations was constructed by Bo\.zejko and Speicher \cite{bozejkospeicherCMP1991}.
%\end{example}

%\begin{example} Let us assume that, for a subset $A\subset T^2$ as
%above, we have additional strict order `$<$', i.e., for all $(s,t)\in T^{(2)}$,
% either $t<s$ or $s<t$. Let $B\in \mathcal B(T)$ be such that $\sigma(B)>0$ and $\sigma(T\setminus B)>0$. Set
%$$ Q(s, t):= \begin{cases} q & \text{if} \ s, t\in B\ \text{and}\  s<t \\
%\overline{q} & \text{if} \ s, t\in B\ \text{and}\  t<s \\
%0 & \text{otherwise}
%\end{cases}.$$
%Clearly the function $Q$ satisfies our assumptions.
%\end{example}

 Let us consider an operator $\Psi$ which transforms a
measurable function $f^{(2)}: T^{(2)}\rightarrow {\mathbb C}$ into the function
\begin{equation}\label{tye654i}
(\Psi f^{(2)})(s,  t):= Q(s, t) f^{(2)}(t, s),\quad (s,t)\in T^{(2)}.
\end{equation}

Analogously to $\td$, we define, for $n\ge3$,
\[
\tn:=\big\{(t_1, \dots , t_n)\in T^n : \text{$(t_i, t_j)\notin E$ for all  $1\leq i < j \leq n$} \big\}.
\]
It is clear that $\sigma^{\otimes n}(T^n\setminus \tn)=0$. The
operator $\Psi$ can be extended to a transformation of functions $\fn : \tn \to \mathbb C$ by
setting, for $k\in \{1,\dots,n-1\}$,
\begin{equation}\label{operators Psi_j}
(\Psi_k\fn)(t_1, \dots , t_n):= Q(t_k, t_{k+1}) \fn(t_1, \dots ,
t_{k-1}, t_{k+1}, t_{k}, t_{k+2}, \dots , t_n ).
\end{equation}

Let $\mathcal H:=L^2(T\to\mathbb C,\sigma)$ be the complex $L^2$-space over $T$. We agree that the scalar product $(\cdot,\cdot)_{\mathcal H}$ is antilinear in the first dot and linear in the second.
For $n\ge2$, the $n$th tensor power of $\mathcal H$, denoted by $\mathcal H^{\otimes n}$, can be identified with the complex $L^2$-space $L^2(T^{(n)}\to\mathbb C,\sigma^{\otimes n})$. Each $\Psi_k$ is a  contraction in $\hcm$.
The following trivial lemma shows that the operators $\Psi_k$  are  self-adjoint and satisfy the
braid relations.

\begin{lemma}\label{ioho90u8y8}
The operators $\Psi_k$ satisfy:
\begin{align}
\Psi_k^*&=\Psi_k,\notag \\
\Psi_k\Psi_l & =  \Psi_l\Psi_k \quad \text{if }|k-l|\ge 2,\notag\\
\Psi_k\Psi_{k+1}\Psi_k & =   \Psi_{k+1}\Psi_{k}\Psi_{k+1}.\label{Yang-Baxter}
\end{align}
\end{lemma}

Let $S_n$ denote the symmetric group on $\{1,\dots,n\}$.
Represent  a permutation $\pi\in S_n$ as an arbitrary  product of adjacent transpositions,
\begin{equation}\label{representation of permutation}
\pi = \pi_{j_1}\dotsm \pi_{j_m},
\end{equation}
 where
$\pi_j:=(j, j+1)\in S_n$ for $1\leq j \leq n-1$.
 A permutation $\pi
\in S_n$ can be represented (not in a unique way, in general) as a
reduced product of a minimal number of adjacent  transpositions,  i.e., in the form  \eqref{representation of permutation} with a minimal $m$.
This number $m$ is
then called the length of $\pi$, denoted by $|\pi|$.  It is well known that $|\pi|$ is equal to the number of inversions of $\pi$, i.e., the number of $1\le i< j\le n$ such that $\pi(i)>\pi(j)$.

The mapping $
\pi_k\mapsto \Psi_{\pi_k}:=\Psi_k
$ can be multiplicatively extended
to $S_n$ by setting
\begin{equation}\label{operators Psi_pi}
S_n\ni\pi \mapsto \Psi_{\pi}:=\Psi_{j_1}\dotsm \Psi_{j_m}.
\end{equation}
Although
representation (\ref{representation of permutation}) of $\pi\in S_n$ in a reduced form is not unique,
the formulas (\ref{Yang-Baxter}) yield that the extension
(\ref{operators Psi_pi}) is well defined, i.e., it does not depend on
the representation. (This fact also follows from the proof of Proposition~\ref{R_pi(f) formula on tensors} below.)

We will use the notations $\mathbf t^{(n)}:=(t_1,\dots,t_n)\in T^{(n)}$, $\mathbf t_\pi^{(n)}:=(t_{\pi(1)},\dots,t_{\pi(n)})$ for $\pi\in S_n$.

\begin{prop}
\label{R_pi(f) formula on tensors}
For each $\pi\in S_n$ and $f^{(n)}\in\mathcal H^{\otimes n}$, we have
\begin{equation}\label{R_pi(f)}
(\Psi_\pi f^{(n)})(\mathbf t^{(n)})=Q_{\pi^{-1}}(\mathbf t^{(n)})f^{(n)}(\mathbf t^{(n)}_\pi),\end{equation}
where
\begin{equation}\label{R_pi_wzor}
Q_{\pi}(\mathbf t^{(n)}):=\!\!\! \prod_{\substack{ 1\le i < j \le n\\[1mm]
\pi(i) > \pi(j)}}
\!\!\! Q(t_i, t_j),\quad \mathbf t^{(n)}\in T^{(n)}.\end{equation}
\end{prop}

For $n\ge 2$, we  define an operator $\pn$ on
$\hcm$ by
\begin{equation}\label{Q-symmetrization}
{\pn}:= \frac{1}{n!}\sum_{\pi\in S_n}\Psi_{\pi}.
\end{equation}
The operator ${\pn}$ is  a self-adjoint contraction in $\mathcal H^{\otimes n}$, since so are the operators $\Psi_k$.

The following result is a special case of Theorem 1.1 in \cite{bozejkospeicherMathAnn1994}.

\begin{thm}[ \cite{bozejkospeicherMathAnn1994}] For each $n\ge 2$, we have $\mathcal P_n\ge0$.
\end{thm}

For any $f^{(n)},g^{(n)}\in \mathcal H^{\otimes n}$, we define
\begin{equation}\label{frte6u4e}
( f^{(n)},g^{(n)})_{\mathcal F_n(\mathcal H)}:=
( \mathcal P_n f^{(n)},g^{(n)})_{\mathcal H^{\otimes n}}\,.
\end{equation}
We consider the factor space
$$\mathcal F_n(\mathcal H):=\mathcal H^{\otimes n}\big/\big\{
f^{(n)}\in \mathcal H^{\otimes n}: ( f^{(n)},f^{(n)})_{\mathcal F_n(\mathcal H)}=0\big\},$$
and define a scalar product on $\mathcal F_n(\mathcal H)$ by \eqref{frte6u4e}.

Below, for a bounded linear operator $L$ in a Hilbert space $\mathfrak H$, we denote by $\Ker(L)$ and $\Ran(L)$ the kernel of $L$ and the range of $L$, respectively. Recall that $\Ker(L)$ is a closed linear subspace of $\mathfrak H$ and, if $L$ is self-adjoint,
$$\mathfrak H=\Ker(L)\oplus \overline{\Ran(L)},$$
where $ \overline{\Ran(L)}$ denotes the closure of the linear subspace $\Ran(L)$. The following lemma only uses the fact that $\mathcal P_n\ge0$.

\begin{lemma}\label{pn has no kernel}
 (i) We have
$$\big\{
f^{(n)}\in \mathcal H^{\otimes n}: ( f^{(n)},f^{(n)})_{\mathcal F_n(\mathcal H)}=0\big\} =\Ker (\mathcal P_n).$$

(ii) For each $f^{(n)}\in \overline{{\Ran}(\pn)}$, $f^{(n)}\ne0$,
$$( f^{(n)},f^{(n)})_{\mathcal F_n(\mathcal H)}>0.$$
\end{lemma}

By Lemma \ref{pn has no kernel}, we can identify $\mathcal F_n(\mathcal H)$ with the set $\overline{{\Ran}(\pn)}$ equipped with scalar product \eqref{frte6u4e}.

The result below follows from Theorem~2 and Remark~4 in \cite{jorgensenproskurinsamoilenkoPacificJM}.

\begin{thm}[\cite{jorgensenproskurinsamoilenkoPacificJM}]\label{f7ir5i4i}
We have
\begin{equation}\label{kernel of Pn}
\Ker({\pn}) = \overline{\sum_{k=1}^{n-1} \Ker(\mathbf 1+\Psi_k)},
\end{equation}
i.e., the kernel of $\mathcal P_n$ is equal to the closure of the linear span of the subspaces\linebreak $ \Ker(\mathbf 1+\Psi_k)$, $k=1,\dots,n-1$.
\end{thm}

We  will now give an explicit description of the space $\mathcal F_n(\mathcal H)=\overline{{\Ran}(\pn)}$. We denote
\begin{equation}\label{utfr7ro5t}
\Theta :=\big \{(s,t)\in {\td} : |Q(s,
t)|=1\big\},\quad \Theta':= T^{(2)}\setminus\Theta=\big\{(s,t)\in {\td} : |Q(s,
t)|<1\big\}.\end{equation}

\begin{thm}\label{main theorem for range of pn}
 The space $\mathcal F_n(\mathcal H)=\overline{{\Ran}(\pn)}$ is equal (as a set) to the subspace of $\hcm$
 consisting of all $f^{(n)}\in \hcm$ that are $Q$-quasisymmetric, i.e., formula \eqref{Q-symmetric functions} holds for each $k\in\{1,\dots,n-1\}$ and
for $\sigma^{\otimes n}$-a.a.\ $(t_1,\dots,t_n)\in T^{(n)}$ such that $|Q(t_k,t_{k+1})|=1$, i.e., for $\sigma^{\otimes n}$-a.a.\ $(t_1,\dots,t_n)\in T^{(n)}_k$, where
\begin{equation}\label{yree8wfw7i}
T^{(n)}_{k}:=\left\{(t_1, \dots , t_n)\in \tn: (t_k,t_{k+1})\in\Theta \right\}.\end{equation}

\end{thm}

\subsection{Orthogonal projection onto $\overline{\Ran(\pn)}$.}

We will now describe the orthogonal projection $\mathbb P_n$ of $\mathcal H^{\otimes n}$ onto $\overline{\Ran(\pn)}=\mathcal F_n(\mathcal H)$.
For this purpose, we define a function
$$R(s,t):=\begin{cases}
Q(s,t),&\text{if }(s,t)\in\Theta,\\
0,&\text{if }(s,t)\in\Theta'.\end{cases}$$
Observe that
$$|R(s,t)|=\begin{cases}
1,&\text{if }(s,t)\in\Theta,\\
0,&\text{if }(s,t)\in\Theta',\end{cases}$$
and that the function $R$ is Hermitian. Hence, for each $\pi\in S_n$, similarly to the operator $\Psi_\pi:\mathcal H^{\otimes n}\to \mathcal H^{\otimes n}$ defined in subsec.~\ref{ydr6ur5ei} for the function $Q(s,t)$, we may define an operator $\Phi_\pi:\mathcal H^{\otimes n}\to \mathcal H^{\otimes n}$ for the function $R(s,t)$. By Proposition~\ref{R_pi(f) formula on tensors}, we get
\begin{equation}\label{rdu6ed6e}
(\Phi_\pi f^{(n)})(\mathbf t^{(n)})=R_{\pi^{-1}}(\mathbf t^{(n)})f^{(n)}(\mathbf t^{(n)}_\pi),\end{equation}
where
\begin{equation}\label{utf7ir}
R_{\pi}(\mathbf t^{(n)}):=\!\!\! \prod_{\substack{ 1\le i < j \le n\\[1mm]
\pi(i) > \pi(j)}}
\!\!\! R(t_i, t_j),\quad \mathbf t^{(n)}\in T^{(n)}.\end{equation}

Let $\pi\in S_n$ and let $\mathbf t^{(n)}\in T^{(n)}$  be such that, for some $1\le i<j\le n$, we have $\pi(i)>\pi(j)$ and $(t_i,t_j)\in\Theta'$. Then, it follows from  \eqref{utf7ir} that $R_{\pi}(\mathbf t^{(n)})=0$. Otherwise, i.e., if such $i$ and $j$ do not exist, we get $|R_{\pi}(\mathbf t^{(n)})|=1$.

Given $\mathbf t^{(n)}\in \tn$, we define a splitting
$$S_n=S_n^1(\mathbf t^{(n)})\sqcup S_n^0(\mathbf t^{(n)})$$
 of the set $S_n$ into two disjoint subsets:
\begin{align}
S_n^1(\mathbf t^{(n)}):&= \{\pi\in S_n: |{R}_{\pi^{-1}}(\mathbf t^{(n)})|=1\},\notag\\
S_n^0(\mathbf t^{(n)}):&=
\{\pi\in S_n: |{R}_{\pi^{-1}}(\mathbf t^{(n)})|=0\}.\label{uyr75er5}
\end{align}

Let $c_n(\mathbf t^{(n)}):=|S_n^1(\mathbf t^{(n)})|$ denote the cardinality.
We define an operator $\mathbb P_n:\hcm \to \hcm$ by setting, for each $f^{(n)}\in\hcm$,
\begin{align}
(\mathbb P_n f^{(n)})(\mathbf t^{(n)})
:&= \frac{1}{c_n(\mathbf t^{(n)})} \sum_{\pi\in S_n^1(\mathbf t^{(n)})}
(\Phi_{\pi}f^{(n)})(\mathbf t^{(n)})\notag\\
& = \frac{1}{c_n(\mathbf t^{(n)})} \sum_{\pi\in S_n^1(\mathbf t^{(n)})}
{R}_{\pi^{-1}}(\mathbf t^{(n)})f^{(n)}(\mathbf t^{(n)}_{\pi}).
\label{R-symmetrization}\end{align}

\begin{thm}\label{orthogonal projection}
For each $n\ge2$, the operator
$\mathbb P_n$ is the orthogonal projection of $ \hcm$ onto $ \overline{\Ran(\pn)}=\mathcal F_n(\mathcal H)$.
\end{thm}

The  corollary below is a straightforward consequence of Theorem~\ref{orthogonal projection}.

\begin{cor}\label{dryde6644e36}
For each $n\ge2$,
$$\mathbb P_n\mathcal P_n=\mathcal P_n\mathbb P_n=\mathcal P_n.$$\end{cor}

We will also need the following result about the operators $\mathbb P_n$, which  follows from Theorem~\ref{orthogonal projection} and its proof.

\begin{cor}\label{uyt67ct}
For each $n\ge2$ and $k\in\{1,\dots,n-1\}$, we have
\begin{equation}\label{R_{n-k}}
\mathbb P_n=\mathbb P_n(\mathbb P_k\otimes\mathbb P_{n-k}).
\end{equation}
Here we denote by $\mathbb P_1:=\mathbf 1$ the identity operator in $\mathcal H$.
\end{cor}

\begin{rem} For $f^{(n)}\in\mathcal F_n(\mathcal H)$
and $g^{(m)}\in\mathcal F_m(\mathcal H)$, we may define a {\it $Q$-quasisymmetric tensor product} of $f^{(n)}$ and $g^{(m)}$ by
$$f^{(n)}\circledast g^{(m)}:=\mathbb P_{n+m}(f^{(n)}\otimes g^{(m)}).$$
Then Corollary \ref{uyt67ct} implies that the $Q$-quasisymmetric tensor product $\circledast$ is associative.
\end{rem}

\subsection{Creation and annihilation operators and their $Q$-commutation relations}

Recall that we have defined complex Hilbert spaces $\mathcal F_n(\mathcal H)$ for $n\ge2$. Let also $\mathcal F_1(\mathcal H):=\mathcal H$ and  $\mathcal F_0(\mathcal H):=\mathbb C$.
We define a $Q$-deformed Fock space to be the Hilbert space
$$
\Fhc := \bigoplus_{n=0}^{\infty} {\Fnhc}\, n!\,.
$$
Thus, every $ f\in \Fhc$ is represented as $f=(f^{(n)})_{n=0}^{\infty}$, where $f^{(n)}\in \Fnhc$, and the norm of $f$ is given by
\[
\|f\|_{\Fhc}^2:=\sum_{n=0}^{\infty}\|f^{(n)}\|_{\mathcal F_n(\mathcal H)}^2\,n!\,.
\]
The vector $\Omega:=(1,0,0,\dots)$ is called the vacuum.

Let $\Fhcfin \subset \Fhc$ be the subspace consisting of all finite sequences of the form $f=(f^{(0)}, f^{(1)}, \dots , f^{(k)}, 0, 0, \dots)$ for some $k\in\N$. The subspace $\Fhcfin$ is evidently dense in  $\mathcal F(\mathcal H)$.

For each $h\in\mathcal H$, we define a creation operator
$a^+(h): \Fhcfin\to\Fhcfin$ by setting
\begin{equation}\label{uf7re7ic}
a^+(h)\Omega:=h,\qquad a^+(h)f^{(n)}:=\mathbb P_{n+1}(h\otimes f^{(n)}),\quad f^{(n)}\in\mathcal F_n(\mathcal H),\ n\in\mathbb N.\end{equation}
The domain of the adjoint operator of $a^+(h)$ in $\mathcal F(\mathcal H)$ contains $\Fhcfin$, and furthermore the annihilation operator $a^-(h):=(a^+(h))^*\restriction \Fhcfin$ also maps $\Fhcfin$ into itself.

The following proposition gives an explicit form of the action of the annihilation operator.

\begin{prop}\label{gdrd6i} For each $h\in\mathcal H$, we have $a^-(h)\Omega=0$, $a^-(h)g=(h,g)_{\mathcal H}$ for $g\in\mathcal H$, and
\begin{align}
&(a^-(h)f^{(n)})(t_1,\dots,t_{n-1})\notag\\
&\quad=\sum_{k=1}^n \mathbb P_{n-1}\left[\int_T \overline{h(s)}\left(\prod_{i=1}^{k-1}Q(s,t_i)\right)f^{(n)}(t_1,\dots,t_{k-1},s,t_k,\dots,t_{n-1})\,\sigma(ds)\right]
\label{d6ueu6eu6}\end{align}
for any $f^{(n)}\in\mathcal F_n(\mathcal H)$, $n\ge2$.
In formula \eqref{d6ueu6eu6}, the operator $\mathbb P_{n-1}$ acts on the function of $t_1,\dots,t_{n-1}$ variables. Furthermore, for any
$g^{(n)}\in\hcm$, $n\ge2$,
\begin{align}
&\big(a^-(h)\mathbb P_n g^{(n)}\big)(t_1,\dots,t_{n-1})\notag\\
&\quad=\sum_{k=1}^n \mathbb P_{n-1}\left[\int_T \overline{h(s)}\left(\prod_{i=1}^{k-1}Q(s,t_i)\right)g^{(n)}(t_1,\dots,t_{k-1},s,t_k,\dots,t_{n-1})\,\sigma(ds)\right].\label{ty456ew}
\end{align}
\end{prop}

For $t\in T$, we now informally define creation and annihilation operators at point $t$, denoted by $\partial_t^\dag$ and $\partial_t$, respectively. A rigorous meaning to these operators is given through smearing them with functions $h\in\mathcal H$:
\begin{equation}\label{yre6ueu4e}
a^{+}(h)=\int_T\sigma(dt)\,h(t)\prd_t, \quad a^-(h)=\int_T\sigma(dt)\,\overline{h(t)}\,\pr_t.
\end{equation}
So we have the following informal equalities:
\begin{align*}
\partial_t^\dag f^{(n)}&=\mathbb P_{n+1}(\delta_t\otimes f^{(n)}),\\
\partial_t f^{(n)}(t_1,\dots,t_{n-1})&=\sum_{k=1}^n \mathbb P_{n-1}\left[ \left(\prod_{i=1}^{k-1}Q(t,t_i)\right)f^{(n)}(t_1,\dots,t_{k-1},t,t_k,\dots,t_{n-1})\right],
\end{align*}
where $\delta_t$ denotes the delta function at $t$.

Using \eqref{uf7re7ic} and Corollary~\ref{uyt67ct}, we see that, for any $g,h\in\mathcal H$ and $f^{(n)}\in\mathcal F_n(\mathcal H)$,
\begin{equation}\label{lyuefgutwe}
a^+(g)a^+(h)f^{(n)}:=\mathbb P_{n+2}(g\otimes h\otimes f^{(n)}).\end{equation}
In view of \eqref{yre6ueu4e} and \eqref{lyuefgutwe}, for each $\varphi^{(2)}\in\mathcal H^{\otimes 2}$, we can naturally define an operator
$$\int_{T^2}\sigma(ds)\,\sigma(dt)\, \varphi^{(2)}(s,t)\,\partial_s^\dag\partial_t^\dag: \Fhcfin\to\Fhcfin $$
by setting
\begin{equation}\label{rdrt}
 \int_{T^2}\sigma(ds)\,\sigma(dt)\, \varphi^{(2)}(s,t)\,\partial_s^\dag\partial_t^\dag f^{(n)}:=\mathbb P_{n+2}(\varphi^{(2)}\otimes f^{(n)})\end{equation}
for $f^{(n)}\in\mathcal F_n(\mathcal H)$. In particular, choosing $\varphi^{(2)}= g\otimes h$ with $g,h\in\mathcal H$, we get
$$ \int_{T^2}\sigma(ds)\,\sigma(dt)\, g(s)h(t)\,\partial_s^\dag\partial_t^\dag=
a^+(g)a^+(h).$$

\begin{rem}
Note that we also  accept the natural formula
 \begin{equation}\label{hyyt9}
 \int_{T^2}\sigma(ds)\,\sigma(dt)\, \varphi^{(2)}(s,t)\,\partial_t^\dag\partial_s^\dag= \int_{T^2}\sigma(ds)\,\sigma(dt)\, \varphi^{(2)}(t,s)\,\partial_s^\dag\partial_t^\dag.\end{equation}
\end{rem}

Similarly, using also Proposition~\ref{gdrd6i}, we may define, for each $\varphi^{(2)}\in\mathcal H^{\otimes 2}$, linear operators
\begin{align*}
&\int_{T^2}\sigma(ds)\,\sigma(dt)\, \varphi^{(2)}(s,t)\,\partial_s\partial_t: \Fhcfin\to\Fhcfin,\\
&\int_{T^2}\sigma(ds)\,\sigma(dt)\, \varphi^{(2)}(s,t)\,\partial_s^\dag\partial_t: \Fhcfin\to\Fhcfin.
\end{align*}
Note that
\begin{align}
\left(\int_{T^2}\sigma(ds)\,\sigma(dt)\, \varphi^{(2)}(s,t)\,\partial_s^\dag\partial_t^\dag\right)^*&= \int_{T^2}\sigma(ds)\,\sigma(dt)\, \overline{\varphi^{(2)}(s,t)}\,\partial_t\partial_s\notag\\
&= \int_{T^2}\sigma(ds)\,\sigma(dt)\, \overline{\varphi^{(2)}(t,s)}\,\partial_s\partial_t.\label{jgfruy7r}
\end{align}
Also, for any $g,h\in\mathcal H$, we denote
$$\int_{T^2}\sigma(ds)\,\sigma(dt)\,g(s)h(t)\partial_s\partial_t^\dag:=a^-(\overline{g}) a^+(h).$$

We will now present the commutation relations for the
creation and annihilation operators.

\begin{thm}[$Q$-CR]\label{Q-CR} The operators $\partial_t^\dag$, $\partial_t$ ($t\in T$) satisfy the (informal) commutations relations \eqref{QCR_0} and \eqref{vyi76}. Rigorously, this means the following: for any $g,h\in\mathcal H$,
\begin{equation}\label{byr75ro5}
\int_{T^2}\sigma(ds)\,\sigma(dt)\,g(s)h(t)\partial_s\partial_t^\dag=\int_T g(t)h(t)\,\sigma(dt)+\int_{T^2}\sigma(ds)\,\sigma(dt)\,g(s)h(t)Q(s,t)\partial_t^\dag \partial_s,
\end{equation}
 and for any function $\varphi^{(2)}\in\mathcal H^{\otimes 2}$ that vanishes a.e.\ in
 $\Theta'$ (see \eqref{utfr7ro5t}),
 \begin{align}
\int_{T^2}\sigma(ds)\,\sigma(dt)\, \varphi^{(2)}(s,t)\,\partial_s^\dag\partial_t^\dag&= \int_{T^2}\sigma(ds)\,\sigma(dt)\, \varphi^{(2)}(s,t)Q(t,s)\,\partial_t^\dag\partial_s^\dag,\label{ft7r54}\\
\int_{T^2}\sigma(ds)\,\sigma(dt)\, \varphi^{(2)}(s,t)\,\partial_s\partial_t&= \int_{T^2}\sigma(ds)\,\sigma(dt)\, \varphi^{(2)}(s,t)Q(t,s)\,\partial_t\partial_s.\label{crte64ue} \end{align}

\end{thm}

We finish this section with several remarks.

\begin{rem} We can naturally identify the diagonal $\Delta:=\{(s,t)\in T^2\mid s=t\}$
with $T$. Denote by $\tilde \sigma$ the measure $\sigma$ on $\Delta$. We may consider $\tilde \sigma$ as a measure on $T^2$ which is equal to zero outside of $\Delta$. Denote
$$\mathfrak G:=L^2(T^2\to\mathbb C,\sigma^{\otimes 2})\cap L^1(T^2\to\mathbb C,\tilde\sigma).$$
In view of \eqref{byr75ro5}, for each $\varphi^{(2)}\in\mathfrak G$, we may define an operator
$$\int_{T^2}\sigma(ds)\,\sigma(dt)\, \varphi^{(2)}(s,t)\,\partial_s\partial_t^\dag: \Fhcfin\to\Fhcfin ,$$
which satisfies
$$\int_{T^2}\sigma(ds)\,\sigma(dt)\,\varphi^{(2)}(s,t)\partial_s\partial_t^\dag=\int_T \varphi^{(2)}(t,t)\,\sigma(dt)+\int_{T^2}\sigma(ds)\,\sigma(dt)\,\varphi^{(2)}(s,t)Q(s,t)\partial_t^\dag \partial_s.$$
\end{rem}

\begin{rem} Denote $B(\varphi):=a^+(\varphi)+a^-(\varphi)$. The family of operators $(B(\varphi))_{\varphi\in\mathcal H}$ can be thought of as a {\it noncommutative Brownian motion} (or a {\it noncommutative Gaussian white noise}). Let $\mathcal P$ denote the complex unital $*$-algebra generated by $(B(\varphi))_{\varphi\in\mathcal H}$, i.e., the algebra of noncommutative polynomials in the variables $B(\varphi)$. We define a vacuum state on $\mathcal P$ by $\tau(p):=(p\Omega,\Omega)_{\mathcal F(\mathcal H)}$, $p\in\mathcal P$. By analogy with the proofs of Theorem~4.4 in \cite{bozejkospeicherMathAnn1994} and Corollary~4.9 in \cite{bozejkolytvynovwysoczanskiCMP2012}, one can prove the following result: the state $\tau$ is tracial (i.e., it satisfies $\tau(p_1p_2)=\tau(p_2p_1)$ for all $p_1,p_2\in\mathcal P$) if and only if the function $Q$ is real-valued, i.e., $Q:T^{(2)}\to[-1,1]$.

\end{rem}

\begin{rem} The results of this section  hold, in particular,  in the case where $\sigma^{\otimes 2}(\Theta')=0$, i.e., when $|Q(s,t)|<1$ for $\sigma^{\otimes 2}$-a.a.\ $(s,t)\in T^2$.
Then, for each $n\ge2$,  the equality $\mathcal F_n(\mathcal H)=\mathcal H^{\otimes n}$ holds (in the sense of sets). Evidently, there are no commutation relations \eqref{ft7r54}, \eqref{crte64ue} in this case. Note also that, if $|Q(s,t)|\leq r <1$ for some number $0<r<1$, then the creation  operators $a^+(h)$ and the annihilation operators $a^-(h)$ ($h\in\mathcal H$) are bounded in $\mathcal F(\mathcal H)$, see Theorem~3.1, (ii) in \cite{bozejkospeicherMathAnn1994}.
\end{rem}

\subsection{Discrete setting: the anyonic exclusion principle}

We will now make several observations about the discrete setting. We may choose $T$ to be a finite or countable set and $\sigma$ to be the counting measure on $T$, i.e., $\sigma(\{t\})=1$ for each $t\in T$. Hence, the space $\mathcal H$ becomes the complex $\ell^2$-space over $T$, i.e.,  $\mathcal H=\ell^2(T\to {\mathbb C})$.
We obviously have $T^{(2)}=T^2$, so that the function $Q(s,t)$ is defined for all $(s,t)\in T^2$.  Thus, we have, in particular, constructed Fock representations of the discrete commutation relations \eqref{ghdtrdeyk} with additional commutation relations between $\partial_s^\dag$, $\partial_t^\dag$ and between $\partial_s$, $\partial_t$  for those pairs $(s,t)\in T^2$ for which $|Q(s,t)|=1$. (Note that, in this case, the operators $\partial_t^\dag$, $\partial_t$ have a rigorous meaning.)

Since the function $Q$ is Hermitian, we  have $Q(t,t)\in\mathbb R$ for each $t\in T$. Hence,  $|Q(t,t)|=1$ if and only if either $Q(t,t)=1$ or $Q(t,t)=-1$. In the first case, we just get the tautological commutation relation $(\partial_t^\dag)^2=(\partial_t^\dag)^2$.
In the second case, we get $(\partial_t^\dag)^2=-(\partial_t^\dag)^2$, so that $(\partial_t^\dag)^2=\partial_t^2=0$. If the latter formulas hold for all $t\in T$, then we may call the corresponding commutation relations the {\it discrete $Q$-CR of fermion type}.

For the discrete  $Q$-CR of fermion type, the operators $\partial_t^\dag$, $\partial_t$ become bounded in $\mathcal F(\mathcal H)$ and have norm equal to 1, see \cite{bozejkospeicherMathAnn1994}, Corollary~3.2 and Remark after it. Hence, for each $h\in\ell^1(T\to {\mathbb C})$,
\begin{equation*}
\|a^+(h)\|=\|a^-(h)\|\le\|h\|_{\ell^1(T\to {\mathbb C})}.\end{equation*}

Let us now assume that $T\subset \mathbb N$ and fix $q\in\mathbb C$, $|q|=1$. We consider the function
\begin{equation*}%\label{Q kernel discrete}
Q(s,t):=
\begin{cases}
q, & {\text{if}} \ s > t  \\
\bar{q}, & {\text{if}} \ s < t \\
-1, & {\text{if}} \ s = t, %\quad \text{for} \quad |q|\leq 1,
\end{cases}
\end{equation*}
The corresponding $Q$-CR describe a {\it discrete anyon system of fermion type}. Note that $|Q(s,t)|=1$ for all $(s,t)\in T^2$, hence $\mathcal P_n=\mathbb P_n$ is the projection of $\mathcal H^{\otimes n}$ onto $\mathcal F_n(\mathcal H)$.

\begin{thm}[Anyonic exclusion principle]\label{Theorem exclusion principle} Consider a discrete anyon system of fermion type.
 Let $m\in\mathbb N$, $m\ge2$. Assume that the parameter $q\in \mathbb C$, $q\ne1$, is an $m$th root of unity, i.e., $q^m=1$.   Then, for any $h\in \mathcal H$, we have
\begin{equation}\label{exclusion principle}
a^+(h)^m=a^-(h)^m=0.
\end{equation}
\end{thm}

\section{Proofs}
In this section we collect the proofs of the results from Section 2.

\begin{proof}[Proof of Proposition \ref{R_pi(f) formula on tensors}]

We start with the following crucial lemma.

\begin{lemma}\label{higft7qwri}
 Let $\rho=\pi_l\eta$ be a reduced representation of a permutation $\rho\in S_n$.  Then
\begin{equation}\label{R_pi(f)_1}
Q_\rho(t_1,\dots,t_n) =  Q(t_{\eta^{-1}(l)},
t_{\eta^{-1}(l+1)})
Q_\eta(t_1,\dots,t_n),\quad (t_1,\dots,t_n)\in T^{(n)}.
\end{equation}
\end{lemma}

\begin{proof}
Let
\[
L_{\rho}:=Q_\rho(t_1,\dots,t_n)=\prod_{\substack{
1\le i < j \le n\\
\rho(i) > \rho(j)}} Q(t_i, t_j),\quad L_{\eta}:= Q_\eta(t_1,\dots,t_n)=\prod_{\substack{
1\le i < j \le n\\
\eta(i) > \eta(j)}} Q(t_i, t_j).
\]
Let $1\le u < v \le n$. We consider the following cases.

\begin{itemize}
\item
If $\eta(u), \eta(v) \notin \{l, l+1 \}$,  then both $\eta(u),
\eta(v)$ are fixed points for  $\pi_l$.
Consequently, $\rho(u)=\eta(u)$ and $\rho(v)=\eta(v)$, so that $\rho(u)>\rho(v)$ if and only if
$\eta(u)>\eta(v)$. Hence, the term $Q(t_u, t_v)$
appears in $L_{\rho}$ if and only if it appears in $L_{\eta}$.

\item
If $\eta(u) \in \{l, l+1 \}$ and $\eta(v)\notin \{l, l+1 \}$, then
$\rho(v)=v\notin \{l, l+1 \}$ and, since $\rho(u)=(\pi_l\eta)(u)\in
\{l, l+1 \}$, the order between $\eta(u)$ and $\eta(v)$ is the same
as between $\rho(u)$ and $\rho(v)$. Thus, the term $Q(t_u, t_v)$
appears in $L_{\rho}$ if and only if it appears in $L_{\eta}$.

\item
The case $\eta(u) \notin \{l, l+1 \}$ and $\eta(v)\in \{l, l+1 \}$
is analogous to the previous one.

\item
Consider the case $\eta(u)=l$ and $\eta(v)=l+1$. Then the term $Q(t_u, t_v)$
does not appear in $L_{\eta}$.  Further, $\rho(u)=(\pi_l\eta)(u)=\pi_l(l)=l+1$ and $\rho(v)=(\pi_l\eta)(v)=\pi_l(l+1)=l$, so that $\rho(u) > \rho(v)$. Hence,
the term $Q(t_u, t_v)$ appears in $L_{\rho}$. But we also have $Q(t_{\eta^{-1}(l)}, t_{\eta^{-1}(l+1)})=Q(t_u,t_v)$ on the right hand side of equality \eqref{R_pi(f)_1}.

\item Finally, consider the case $\eta(u)=l+1$ and $\eta(v)=l$. But then $\rho(u)=(\pi_l\eta)(u)=l$ and $\rho(v)=(\pi_l\eta)(v)=l+1$. Thus, $\eta$ changes the order of the pair $(u,v)$, while $\rho$ does not. Therefore,
$\eta$ has more inversions  than $\rho$: $|\eta|>|\rho|$. But this contradicts the assumption that  $\rho$
is in the reduced form. Thus, this case is impossible.
\end{itemize}
\end{proof}

We will now prove the proposition by induction on the length of a permutation $\pi$.
If $|\pi|=1$, then $\pi=\pi_k$ for some $k\in\{1,\dots,n-1\}$. In this case, the statement  trivially follows from the definition of $\Psi_k$, see \eqref{operators Psi_j}. Assume that the statement holds for each permutation of length $m$. Let $\pi$ be a a permutation of length $m+1$, and let $\pi=\varphi\pi_l$ be a reduced representation of $\pi$. Hence, the length of the permutation $\varphi$ is $m$. Denote $\eta:=\varphi^{-1}$ and  $\rho:=\pi^{-1}$, so that $\rho=\pi_l\eta$. Then, for each $f^{(n)}\in\hcm$, by using the induction's assumption and Lemma~\ref{higft7qwri}, we get
\begin{align*}
&(\Psi_\pi f^{(n)})(t_1,\dots,t_n)=(\Psi_\varphi\Psi_l f^{(n)})(t_1,\dots,t_n)\\
&\quad=Q_\eta(t_1,\dots,t_n)(\Psi_l f^{(n)})(t_{\varphi(1)},\dots,t_{\varphi(n)})\\
&\quad=Q_\eta(t_1,\dots,t_n)Q(t_{\varphi(l)},t_{\varphi(l+1)}) f^{(n)}(t_{\varphi(1)},\dots,t_{\varphi(l+1)},t_{\varphi(l)},\dots,t_{\varphi(n)})\\
&\quad=Q_\rho(t_1,\dots,t_n)f^{(n)}(t_{\pi(1)},\dots,t_{\pi(n)}).\end{align*}
\end{proof}

\begin{proof}[Proof of Lemma \ref{pn has no kernel}]

(i) Since $\pn$ is self-adjoint and $\pn\geq 0$, we can write $\pn=(\sqrt{\pn})^2$. Let $f^{(n)}\in\hcm$ be such that
$$0=( f^{(n)},f^{(n)})_{\mathcal F_n(\mathcal H)}=\|\sqrt{\pn} f^{(n)}\|^2_{\hcm}.$$
Hence, $f^{(n)}\in\Ker(\sqrt{\mathcal P_n})$. But $\Ker \sqrt{\pn}\subset \Ker \pn$, which implies
$$ \big\{
f^{(n)}\in \mathcal H^{\otimes n}\mid ( f^{(n)},f^{(n)})_{\mathcal F_n(\mathcal H)}=0\big\}\subset \Ker (\mathcal P_n).$$
The inverse inclusion trivially follows  from \eqref{frte6u4e}.

(ii) Let $f^{(n)}\in \overline{{\Ran}(\pn)}$ be such that $( f^{(n)},f^{(n)})_{\mathcal F_n(\mathcal H)}=0$. By part (i), $f^{(n)}\in\Ker (\mathcal P_n)$. But $\overline{{\Ran}(\pn)}\perp \Ker (\mathcal P_n)$. Hence,
$\overline{{\Ran}(\pn)}\cap \Ker (\mathcal P_n)=\{0\}$, and so $f^{(n)}=0$.
\end{proof}

\begin{proof}[Proof of Theorem~\ref{main theorem for range of pn}]

Using (\ref{kernel of Pn}), we have
\begin{equation}\label{tydeu645w38}
\overline{{\Ran}(\pn)} = \left(\sum_{k=1}^{n-1} \Ker(\mathbf 1+\Psi_k)\right)^{\perp} = \bigcap_{k=1}^{n-1} \Ker(\mathbf 1+\Psi_k)^{\perp}=\bigcap_{k=1}^{n-1}\overline{{\Ran}(\mathbf 1+\Psi_k)}.
\end{equation}
For $l\in \N$ and $k\in\{1,\dots,n-1\}$, we denote
\[
T^{(n)}_{k,l}:=\left\{(t_1, \dots , t_n)\in \tn : \frac{l-1}{l}\leq |Q(t_k, t_{k+1})|<  \frac{l}{l+1}\right\}
\]
and  recall the definition of $T_k^{(n)}$, see \eqref{yree8wfw7i}.
Then, for each $k\in\{1,\dots,n-1\}$, we have the orthogonal decomposition
\begin{equation}\label{L2 decomposition}
\mathcal H^{\otimes n}=\left(\bigoplus_{l=1}^{\infty} L^2(T^{(n)}_{k,l}\to\mathbb C, \sigma^{\otimes n})\right) \oplus L^2(T^{(n)}_{k}\to\mathbb C, \sigma^{\otimes n}).
\end{equation}
Each of the spaces on the right-hand side of \eqref{L2 decomposition} is invariant for the operator $\mathbf 1+\Psi_k$.  On each space $L^2(T^{(n)}_{k,l}\to\mathbb C, \sigma^{\otimes n})$, the norm of the operator $\Psi_k$ is bounded by $\frac{l}{l+1}<1$. Hence, the operator $\mathbf 1+\Psi_k$ is invertible in this space. Therefore the kernel of the operator $\mathbf 1+\Psi_k$ restricted to $\displaystyle  L^2(T^{(n)}_{k,l}\to\mathbb C, \sigma^{\otimes n})$ is trivial:
\[
\Ker(\mathbf 1+\Psi_k) \cap L^2(T^{(n)}_{k,l}\to\mathbb C, \sigma^{\otimes n}) =\{0\} \quad \text{for each } l\in \N.
\]

Let $ f^{(n)}\in L^2(T^{(n)}_{k}\to\mathbb C, \sigma^{\otimes n})$. Consider the decomposition $f^{(n)}=f^{(n)}_{k,+} + f^{(n)}_{k,-}$ with
$$
f^{(n)}_{k,\pm}(t_1, \dots , t_n):=\frac{1}{2}\big[f^{(n)}(t_1, \dots , t_n) \pm Q(t_k,t_{k+1})f^{(n)}(t_1, \dots, t_{k+1},t_k,\dots , t_n)\big].
$$
One can easily see that $f^{(n)}_{k,+}$ and $f^{(n)}_{k,-}$ are orthogonal and  $f^{(n)}_{k,+}\in \Ran(\mathbf 1+\Psi_k)$. Hence $f^{(n)}_{k,-}\in\Ker(\mathbf 1+\Psi_k)$.
Therefore, the orthogonal projection of $L^2(T^{(n)}_{k}\to\mathbb C, \sigma^{\otimes n})$ onto $\Ker(\mathbf 1+\Psi_k)$, denoted by  $D^{(n)}_k$, is given by
\[
(D^{(n)}_kf^{(n)})(t_1, \dots , t_n) = \frac{1}{2}\big[f^{(n)}(t_1, \dots , t_n) -Q(t_k,t_{k+1})f^{(n)}(t_1, \dots, t_{k+1},t_k,\dots , t_n) \big] .
\]
Hence, the orthogonal projection of $\mathcal H^{\otimes n}$ onto
$\Ker(\mathbf 1+\Psi_k)$, denoted by $E^{(n)}_k$, is given by
\begin{align*}
&(E^{(n)}_kf^{(n)})(t_1, \dots , t_n)\\
&\quad = \frac{1}{2}\chi_{T_k^{(n)}}(t_1,\dots,t_n)\big[f^{(n)}(t_1, \dots , t_n) -Q(t_k,t_{k+1})f^{(n)}(t_1, \dots, t_{k+1},t_k,\dots , t_n) \big],
\end{align*}
where $\chi_{A}$ denotes the indicator function of a set $A$. Therefore, the orthogonal projection
 of $\mathcal H^{\otimes n}$  onto
$ \Ker(\mathbf 1+\Psi_k)^\perp=\overline{\Ran(\mathbf 1+\Psi_k)}$, denoted by $F_k^{(n)}$, is given by
\begin{align*}
&(F_k^{(n)}f^{(n)})(t_1, \dots , t_n)= \chi_{T^{(n)}\setminus T_k^{(n)}}(t_1, \dots , t_n)f^{(n)}(t_1, \dots , t_n)\\
&\quad + \frac{1}{2}\chi_{T_k^{(n)}}(t_1,\dots,t_n)\big[f^{(n)}(t_1, \dots , t_n) +Q(t_k,t_{k+1})f^{(n)}(t_1, \dots, t_{k+1},t_k,\dots , t_n) \big].
\end{align*}

Thus, the set $\overline{\Ran(\mathbf 1+\Psi_k)}$ consists of all functions
from $\hcm$   that are $Q$-quasi\-symmetric in the $t_k,t_{k+1}$-variables on the set $T^{(n)}_k$, i.e., for $\sigma^{\otimes n}$-a.a\ $(t_1,\dots,t_n)\in T^{(n)}_k$, equality \eqref{Q-symmetric functions} holds. From here and formula \eqref{tydeu645w38}, the theorem follows.

\end{proof}

\begin{proof}[Proof of Theorem \ref{orthogonal projection}] We start with the following lemma.

\begin{lemma}\label{relations L1}
(i) Let  $\mathbf t^{(n)}\in \tn$. Then $\pi\in S_n^1(\mathbf t^{(n)})$ if and only if $\pi^{-1}\in S_n^1(\mathbf t_\pi^{(n)})$.

(ii) Let $\mathbf t^{(n)}\in \tn$, let $\pi \in S_n^1(\mathbf t^{(n)})$, and let
$\nu\in S_n^1(\mathbf t^{(n)}_{\pi})$. Then $\vp:=\pi\nu \in S_n^1(\mathbf t^{(n)})$.

(iii) For each $\mathbf t^{(n)}\in T^{(n)}$ and $\pi\in S_n^1(\mathbf t^{(n)})$, we have $c_n(\mathbf t^{(n)})=c_n(\mathbf t^{(n)}_\pi)$.
\end{lemma}

\begin{proof} (i) By \eqref{utf7ir},
\begin{equation}\label{gt87t5}
R_{\pi^{-1}}(\mathbf t^{(n)})=\overline{R_\pi(\mathbf t_\pi^{(n)})}.\end{equation}
From here the statement follows.

(ii) Assume that $\vp \notin S_n^1(\mathbf t^{(n)})$.
Then there exist $i<j$  such that $\vp^{-1}(i)>\vp^{-1}(j)$ and
$R(t_i, t_j)=0$. Let us consider two cases.

{\it Case 1: $\pi^{-1}(i)>\pi^{-1}(j)$}. But then \eqref{utf7ir} implies that $R_{\pi^{-1}}(\mathbf t^{(n)})=0$, hence $\pi\not \in  S_n^1(\mathbf t^{(n)})$, which is a contradiction.

{\it Case 2: $\pi^{-1}(i)<\pi^{-1}(j)$}. We then have
$$\nu^{-1}(\pi^{-1}(i))=\vp^{-1}(i)>\vp^{-1}(j)=\nu^{-1}(\pi^{-1}(j)).$$  By \eqref{utf7ir},
$$R_{\nu^{-1}}(\mathbf t_\pi^{(n)}):=\!\!\! \prod_{\substack{ 1\le a < b \le n\\[1mm]
\nu^{-1}(a) > \nu^{-1}(b)}}
\!\!\! R(t_{\pi(a)}, t_{\pi(b)}).$$
Choose $a=\pi^{-1}(i)$ and $b=\pi^{-1}(j)$. Then $a<b$, $\nu^{-1}(a)>\nu^{-1}(b)$, and
$$R(t_{\pi(a)}, t_{\pi(b)})=R(t_i,t_j)=0.$$
Therefore, $R_{\nu^{-1}}(\mathbf t_\pi^{(n)})=0$, which implies $\nu\not\in S_n^1(\mathbf t^{(n)}_{\pi})$.  This is again a contradiction. Thus, we must have $\vp \in S_n^1(\mathbf t^{(n)})$.

(iii) By part (ii), if $\nu\in S_n^1(\mathbf t_\pi^{(n)})$, then $\pi\nu\in S_n^1(\mathbf t^{(n)})$. Hence, $c_n(\mathbf t_\pi^{(n)})\le c_n(\mathbf t^{(n)})$.
On the other hand, by part (i),  $\pi^{-1}\in S_n^1(\mathbf t_\pi^{(n)})$. Hence, by part (i), if $\mu\in S_n^1(\mathbf t^{(n)})$ then $\pi^{-1}\mu\in S_n^1(\mathbf t^{(n)}_\pi)$. Hence, $c_n(\mathbf t^{(n)})\le c_n(\mathbf t_\pi^{(n)})$.
\end{proof}

 We first show that the operator $\mathbb P_n$ is self-adjoint.
By \eqref{rdu6ed6e}--\eqref{R-symmetrization}, we can write the operator $\mathbb P_n$ in the form
$$(\mathbb P_n f^{(n)})(\mathbf t^{(n)})
= \frac{1}{c_n(\mathbf t^{(n)})} \sum_{\pi\in S_n}
{R}_{\pi^{-1}}(\mathbf t^{(n)})f^{(n)}(\mathbf t^{(n)}_{\pi}).$$
Hence using Lemma~\ref{relations L1}, (iii) and \eqref{gt87t5}, we get, for any $f^{(n)},g^{(n)}\in\mathcal H^{\otimes n}$,
\begin{align}
&(\mathbb P_n f^{(n)},g^{(n)})_{\mathcal H^{\otimes n}}=\sum_{\pi\in S_n}\int_{T^{(n)}}\frac{1}{c_n(\mathbf t^{(n)})}\,\overline{
R_{\pi^{-1}}(\mathbf t^{(n)})f^{(n)}(\mathbf t_\pi^{(n)})}\,g^{(n)}(\mathbf t^{(n)})\,\sigma^{\otimes n}(d\mathbf t^{(n)})\notag\\
&\quad=\sum_{\pi\in S_n}\int_{T^{(n)}}\frac{1}{c_n(\mathbf t_{\pi^{-1}}^{(n)})}\,\overline{
R_{\pi^{-1}}(\mathbf t_{\pi^{-1}}^{(n)})f^{(n)}(\mathbf t^{(n)})}\,g^{(n)}(\mathbf t_{\pi^{-1}}^{(n)})\,\sigma^{\otimes n}(d\mathbf t^{(n)})\notag\\
&\quad=\sum_{\pi\in S_n}\int_{T^{(n)}}\frac{1}{c_n(\mathbf t_{\pi}^{(n)})}\,\overline{
R_{\pi}(\mathbf t_{\pi}^{(n)})f^{(n)}(\mathbf t^{(n)})}\,g^{(n)}(\mathbf t_{\pi}^{(n)})\,\sigma^{\otimes n}(d\mathbf t^{(n)})\notag\\
&\quad=\sum_{\pi\in S_n}\int_{T^{(n)}}\frac{1}{c_n(\mathbf t_{\pi}^{(n)})}R_{\pi^{-1}}(\mathbf t^{(n)})\,\overline{
f^{(n)}(\mathbf t^{(n)})}\,g^{(n)}(\mathbf t_{\pi}^{(n)})\,\sigma^{\otimes n}(d\mathbf t^{(n)})\notag\\
&\quad=\int_{T^{(n)}}\overline{
f^{(n)}(\mathbf t^{(n)})}\,\sum_{\pi\in S_n^1(\mathbf t^{(n)})}\frac{1}{c_n(\mathbf t_{\pi}^{(n)})}R_{\pi^{-1}}(\mathbf t^{(n)})\,g^{(n)}(\mathbf t_{\pi}^{(n)})\,\sigma^{\otimes n}(d\mathbf t^{(n)})\notag\\
&\quad=\int_{T^{(n)}}\overline{
f^{(n)}(\mathbf t^{(n)})}\,\frac{1}{c_n(\mathbf t^{(n)})}\sum_{\pi\in S_n^1(\mathbf t^{(n)})}R_{\pi^{-1}}(\mathbf t^{(n)})\,g^{(n)}(\mathbf t_{\pi}^{(n)})\,\sigma^{\otimes n}(d\mathbf t^{(n)})\notag\\
&\quad=( f^{(n)},\mathbb P_ng^{(n)})_{\mathcal H^{\otimes n}}.
\end{align}
Thus, $\mathbb P_n^*=\mathbb P_n$.

Our next aim is to prove that $\mathbb P_n^2=\mathbb P$, which will imply that $\mathbb P_n$ is an orthogonal projection in $\hcm$.
For $f^{(n)}\in\hcm$, we have, by Lemma~\ref{relations L1}, (ii) and (iii),
\begin{align}
(\mathbb P_n^2 f^{(n)})(\mathbf t^{(n)})&=\frac1{c_n(\mathbf t^{(n)})}\sum_{\pi\in S_n^1(\mathbf t^{(n)})}\frac{1}{c_n(\mathbf t_\pi^{(n)})}\sum_{\nu\in S_n^1(\mathbf t^{(n)}_\pi)} (\Phi_\pi\Phi_\nu f^{(n)})(\mathbf t^{(n)})\notag\\
&=\frac1{c_n(\mathbf t^{(n)})^2}\sum_{\pi\in S_n^1(\mathbf t^{(n)})}\sum_{\nu\in S_n^1(\mathbf t^{(n)}_\pi)} (\Phi_\pi\Phi_\nu f^{(n)})(\mathbf t^{(n)})\notag\\
&=\frac1{c_n(\mathbf t^{(n)})^2}\sum_{\varphi\in S_n^1(\mathbf t^{(n)})}
\sum_{\substack{{\pi\in S_n^1(\mathbf t^{(n)}),\, \nu\in S_n^1(\mathbf t^{(n)}_\pi)}\\\pi\nu=\varphi}}(\Phi_\pi\Phi_\nu f^{(n)})(\mathbf t^{(n)}).\label{tyde6e68}
\end{align}
Let $\varphi\in S_n^1(\mathbf t^{(n)})$ and $\pi\in S_n^1(\mathbf t^{(n)})$.
By  Lemma~\ref{relations L1}, (i), we have $\pi^{-1}\in S_n^1(\mathbf t_\pi^{(n)})$. Hence, by  Lemma~\ref{relations L1}, (ii), we get $\nu:=\pi^{-1}\varphi\in S_n^1(\mathbf t_\pi^{(n)})$. From here and \eqref{tyde6e68} we get:
\begin{equation}\label{uyf7r}
(\mathbb P_n^2 f^{(n)})(\mathbf t^{(n)})=
\frac1{c_n(\mathbf t^{(n)})^2}\sum_{\varphi\in S_n^1(\mathbf t^{(n)})}
\sum_{\pi\in S_n^1(\mathbf t^{(n)})}(\Phi_\pi\Phi_{\pi^{-1}\varphi} f^{(n)})(\mathbf t^{(n)}).
\end{equation}

\begin{lemma}\label{tydf6ir5}
Let $\mathbf t^{(n)}\in T^{(n)}$,  and $\pi\in S_n^1(\mathbf t^{(n)})$, and $\nu \in S_n^1(\mathbf t_\pi^{(n)})$. Then, for each $f^{(n)}\in\hcm$,
\begin{equation}\label{egyfoqegot8}
(\Phi_\pi\Phi_{\nu} f^{(n)})(\mathbf t^{(n)})=(\Phi_{\pi\nu} f^{(n)})(\mathbf t^{(n)})
\end{equation}
\end{lemma}

\begin{proof} We first note that equality \eqref{egyfoqegot8} explicitly means that
\begin{equation}
R_{\pi^{-1}}(\mathbf t^{(n)})R_{\nu^{-1}}(\mathbf t^{(n)}_\pi)f^{(n)}(\mathbf t^{(n)}_{\pi\nu })=R_{\nu^{-1}\pi^{-1}}(\mathbf t^{(n)})f^{(n)}(\mathbf t^{(n)}_{\pi\nu }),\notag
\end{equation}
which is equivalent to the equality
\begin{equation}\label{yutr7i5}
R_{\pi^{-1}}(\mathbf t^{(n)})R_{\nu^{-1}}(\mathbf t^{(n)}_\pi)=R_{\nu^{-1}\pi^{-1}}(\mathbf t^{(n)}).
\end{equation}
Since $\pi\in S_n^1(\mathbf t^{(n)})$, $\nu \in S_n^1(\mathbf t_\pi^{(n)})$, and $\pi\nu\in S_n^1(\mathbf t^{(n)})$, we have
\begin{equation}\label{yutro78t6}
|R_{\pi^{-1}}(\mathbf t^{(n)})|=1,\quad |R_{\nu^{-1}}(\mathbf t^{(n)}_\pi)|=1,\quad |R_{\nu^{-1}\pi^{-1}}(\mathbf t^{(n)})|=1.\end{equation}

We define a Hermitian function $G:T^{(2)}\to\mathbb C$ by
\begin{equation}
G(s,t):=\begin{cases}\label{yuro86ro8}
R(s,t),&\text{if }|R(s,t)|=1,\\
1,&\text{if }R(s,t)=0.\end{cases}\end{equation}
For each $\pi\in S_n$, similarly to the operator $\Psi_\pi$ defined for the function $Q$ and to the operator $\Phi_\pi$ defined for the function $R$, we define an operator $\Gamma_\pi:\hcm\to\hcm$ for the function $G$. Thus,
$$
(\Gamma_\pi f^{(n)})(\mathbf t^{(n)})=G_{\pi^{-1}}(\mathbf t^{(n)})f^{(n)}(\mathbf t^{(n)}_\pi),$$
where
\begin{equation}\label{uygt8ot8}
G_{\pi}(\mathbf t^{(n)}):=\!\!\! \prod_{\substack{ 1\le i < j \le n\\
\pi(i) > \pi(j)}}
\!\!\! G(t_i, t_j),\quad \mathbf t^{(n)}\in T^{(n)}.\end{equation}
For an adjacent  transposition $\pi_j=(j,j+1)$, we denote $\Gamma_j:=\Gamma_{\pi_j}$. By Lemma~\ref{ioho90u8y8}, the operators $\Gamma_j$ satisfy the braid relations. Furthermore, since $|G(s,t)|=1$ for all $(s,t)\in T^{(2)}$, we get $\Gamma_j^2=\mathbf 1$. Using e.g.\ \cite{CM}, we therefore conclude that the operators $\Gamma_\pi$ with $\pi\in S_n$ form a unitary representation of $S_n$, i.e., for any $\pi,\nu\in S_n$, it holds that $\Gamma_\pi\Gamma_\nu=\Gamma_{\pi\nu}$, and in fact, for all  $\mathbf t^{(n)}\in T^{(n)}$,
\begin{equation}\label{yt65060}
G_{\pi^{-1}}(\mathbf t^{(n)})G_{\nu^{-1}}(\mathbf t^{(n)}_\pi)=G_{\nu^{-1}\pi^{-1}}(\mathbf t^{(n)}).
\end{equation}
But if $\mathbf t^{(n)}\in T^{(n)}$,   $\pi\in S_n^1(\mathbf t^{(n)})$, and $\nu \in S_n^1(\mathbf t_\pi^{(n)})$, then formulas \eqref{yutro78t6}--\eqref{yt65060} imply \eqref{yutr7i5}.
\end{proof}

Now formula \eqref{uyf7r} and Lemma~\ref{tydf6ir5} yield the equality
\begin{align}
(\mathbb P_n^2 f^{(n)})(\mathbf t^{(n)})&=
\frac1{c_n(\mathbf t^{(n)})^2}\sum_{\varphi\in S_n^1(\mathbf t^{(n)})}
\sum_{\pi\in S_n^1(\mathbf t^{(n)})}(\Phi_\varphi f^{(n)})(\mathbf t^{(n)})\notag\\
&=
\frac1{c_n(\mathbf t^{(n)})}\sum_{\varphi\in S_n^1(\mathbf t^{(n)})}
(\Phi_\varphi f^{(n)})(\mathbf t^{(n)})=(\mathbb P_n f^{(n)})(\mathbf t^{(n)}).
\label{tydyr}
\end{align}
Thus, $\mathbb P_n$ is an orthogonal projection in $\hcm$.

It remains to prove that $\Ran(\mathbb P_n)=\mathcal F_n(\mathcal H)$. Let $f^{(n)}\in \mathcal F_n(\mathcal H)$. Theorem~\ref{main theorem for range of pn} and the construction of the $\Phi_\pi$ operators imply that, for  $\sigma^{\otimes n}$-a.a.\ $\mathbf t^{(n)}\in T^{(n)}$ and for each $\pi\in  S_n^1(\mathbf t^{(n)})$, we have $(\Phi_\pi f^{(n)})(\mathbf t^{(n)})=f^{(n)}(\mathbf t^{(n)})$. Hence, by \eqref{R-symmetrization}, $\mathbb P_n f^{(n)}=f^{(n)}$, i.e., $f^{(n)}\in \Ran(\mathbb P_n)$.

Finally, we have to prove the inclusion $\Ran(\mathbb P_n)\subset \mathcal F_n(\mathcal H)$. This means that, for any $f^{(n)}\in\hcm$ and  $k\in\{1,\dots,n-1\}$,
\begin{equation}\label{uytr7i5ww}
(\Phi_k\mathbb P_n f^{(n)})(\mathbf t^{(n)})=(\mathbb P_n f^{(n)})(\mathbf t^{(n)})\quad \text{for $\sigma^{\otimes n}$-a.a.\ $\mathbf t^{(n)}\in T_k^{(n)}$}.
\end{equation}
The proof of  \eqref{uytr7i5ww} is similar to the proof of the equality $\mathbb P_n^2=\mathbb P$ (formulas~\eqref{tyde6e68}, \eqref{uyf7r}, and \eqref{tydyr}), so we omit it.
\end{proof}

\begin{proof}[Proof of Corllary \ref{uyt67ct}]
We start with the following lemma

\begin{lemma}
For each $n\in\mathbb N$, we have
\begin{align}
&\mathbb P_{n+1}(\mathbb P_n\otimes \mathbf 1)=\mathbb P_{n+1},\label{fyr65e}\\
&\mathbb P_{n+1}(\mathbf 1\otimes \mathbb P_n)=\mathbb P_{n+1}.\label{gyfg8q7}
\end{align}
\end{lemma}

\begin{proof}
We will only prove equality \eqref{fyr65e}, since the proof of \eqref{gyfg8q7} is similar. For a permutation $\nu\in S_n$, we denote by $\nu\otimes\operatorname{id}$ the permutation from $S_{n+1}$ defined by
$(\nu\otimes\operatorname{id})(i):=\nu(i)$ for $i\in\{1,\dots,n\}$ and $(\nu\otimes\operatorname{id})(n+1):=n+1$.
Analogously to the proof of Theorem \ref{orthogonal projection}, we get, for any $f^{(n+1)}\in\mathcal H^{\otimes(n+1)}$,
\begin{align}
&(\mathbb P_{n+1}(\mathbb P_n\otimes \mathbf 1)f^{(n+1)})(\mathbf t^{(n+1)})\notag\\
&\quad=\frac1{c_{n+1}(\mathbf t^{(n+1)})}\sum_{\pi\in S_{n+1}^1(\mathbf t^{(n+1)})}
\frac1{c_n(t_{\pi(1)},\dots,t_{\pi(n)})}\sum_{\nu\in S_n^1(t_{\pi(1)},\dots,t_{\pi(n)})}(\Phi_{\pi(\nu\otimes\operatorname{id})}f^{(n+1)})(\mathbf t^{(n+1)})\notag\\
&\quad=\frac1{c_{n+1}(\mathbf t^{(n+1)})}\sum_{i=1}^{n+1}
\sum_{\substack{\varphi\in S_{n+1}^1(\mathbf t^{(n+1)})
\\ \varphi(n+1)=i}}\sum_{\substack{\pi\in S_{n+1}^1(\mathbf t^{(n+1)})\\ \pi(n+1)=i\\\nu\in S_n^1(t_{\pi(1)},\dots,t_{\pi(n)})\\\pi(\nu\otimes\operatorname{id})=\varphi}}\frac1{c_n(t_{\pi(1)},\dots,t_{\pi(n)})}(\Phi_\varphi f^{(n+1)})(\mathbf t^{(n+1)}).\label{rrsrs}
\end{align}
Let $\varphi,\pi\in S_{n+1}^1(\mathbf t^{(n+1)})$ be such that $\pi(n+1)=\nu(n+1)=i$. Then $\nu':=\pi^{-1}\varphi\in S_{n+1}^1(\mathbf t^{(n+1)}_\pi)$ and $\nu'(n+1)=n+1$. Therefore, $\nu'=\nu\otimes\operatorname{id}$, where $\nu\in S_n^1(t_{\pi(1)},\dots,t_{\pi(n)})$. Hence, by \eqref{rrsrs},
\begin{align*}
&(\mathbb P_{n+1}(\mathbb P_n\otimes \mathbf 1)f^{(n+1)})(\mathbf t^{(n+1)})\notag\\
&\quad=\frac1{c_{n+1}(\mathbf t^{(n+1)})}\sum_{i=1}^{n+1}
\sum_{\substack{\varphi\in S_{n+1}^1(\mathbf t^{(n+1)})
\\ \varphi(n+1)=i}}
(\Phi_\varphi f^{(n+1)})(\mathbf t^{(n+1)})
\sum_{\substack{\pi\in S_{n+1}^1(\mathbf t^{(n+1)})\\ \pi(n+1)=i}}\frac1{c_n(t_{\pi(1)},\dots,t_{\pi(n)})}.
\end{align*}
Therefore, it is sufficient to prove that, for any $\mathbf t^{(n+1)}\in T^{(n+1)}$ and $i\in\{1,\dots,n+1\}$,
\begin{equation}\label{tyd6ed6}
\sum_{\substack{\pi\in S_{n+1}^1(\mathbf t^{(n+1)})\\ \pi(n+1)=i}}\frac1{c_n(t_{\pi(1)},\dots,t_{\pi(n)})}=1.
\end{equation}

To this end, we denote
$$S_{n,i}^1(\mathbf t^{(n+1)}):=\{\pi\in S_n^1(\mathbf t^{(n+1)}): \pi(n+1)=i\},$$
and let $c_{n+1,i}(\mathbf t^{(n+1)}):=|S_{n,i}^1(\mathbf t^{(n+1)})|$.
We state that, for any $\mathbf t^{(n+1)}\in T^{(n+1)}$ and
 $\pi\in S_{n,i}^1(\mathbf t^{(n+1)})$,
\begin{equation}\label{utf7rro}
c_{n+1,i}(\mathbf t^{(n+1)})=c_n(t_{\pi(1)},\dots,t_{\pi(n)}).
\end{equation}
Indeed, if $\nu\in S_n^1(t_{\pi(1)},\dots,t_{\pi(n)})$, then $\nu\otimes\operatorname{id}\in S_{n+1}^1(\mathbf t_\pi^{(n+1)})$. Therefore, $\pi(\nu\otimes\operatorname{id})\in S_{n+1}^1(\mathbf t^{(n+1)})$ and
$$\big(\pi(\nu\otimes\operatorname{id})\big)(n+1)=\pi(n+1)=i.$$
Hence, $\pi(\nu\otimes\operatorname{id})\in S_{n+1,i}^1(\mathbf t^{(n+1)})$. So
$c_n(t_{\pi(1)},\dots,t_{\pi(n)})\le c_{n+1,i}(\mathbf t^{(n+1)})$. On the other hand, take any $\varphi\in S_{n+1,i}^1(\mathbf t^{n+1})$. Let $\nu':=\pi^{-1}\varphi$. As shown above, $\nu'=\nu\otimes\operatorname{id}$, where $\nu\in S_n^1(t_{\pi(1)},\dots,t_{\pi(n)})$. Hence, $c_{n+1,i}(\mathbf t^{(n+1)})\le c_n(t_{\pi(1)},\dots,t_{\pi(n)})$, and formula  \eqref{utf7rro} is proven. Finally, formula \eqref{utf7rro}  implies \eqref{tyd6ed6}.
\end{proof}

\begin{lemma}\label{tyr75i45ir} For $k\in\mathbb N$, we denote by $\mathbf 1_k$ the identity operator in $\mathcal H^{\otimes k}$. Then, for each $n\ge 2$,
\begin{align}
\mathbb P_{n+k}(\mathbb P_n\otimes \mathbf 1_k)&=\mathbb P_{n+k},\label{jigr}\\
\mathbb P_{n+k}(\mathbf 1_k\otimes \mathbb P_n)&=\mathbb P_{n+k}.\label{tre4u6}
\end{align}
\end{lemma}

\begin{proof}We will again only prove the first formula, \eqref{jigr}, the proof of \eqref{tre4u6} being similar.
 We prove \eqref{jigr}  by induction on $k$. For $k=1$, formula \eqref{jigr} becomes  \eqref{fyr65e}. Let $k\ge2$ and assume that formula \eqref{fyr65e} holds for  $k-1$. We then have
\begin{align*}
\mathbb P_{n+k}(\mathbb P_n\otimes \mathbf 1_k)&=\mathbb P_{n+k}(\mathbb P_{n+k-1}\otimes\mathbf 1)(\mathbb P_n\otimes \mathbf 1_{k-1}\otimes\mathbf 1)\\
&=\mathbb P_{n+k}\big[\big(\mathbb P_{n+k-1}(\mathbb P_n\otimes\mathbf 1_{k-1})\big)\otimes\mathbf 1\big]\\
&=\mathbb P_{n+k}(\mathbb P_{n+k-1}\otimes\mathbf 1)=\mathbb P_{n+k}.
\end{align*}
\end{proof}

Using Lemma \ref{tyr75i45ir}, we get, for $n\ge2$ and $k\in\{1,\dots,n-1\}$
$$\mathbb P_n (\mathbb P_k\otimes \mathbb P_{n-k})=\mathbb P_n (\mathbb P_k\otimes\mathbf 1_{n-k})(\mathbf 1_k\otimes\mathbb P_{n-k})=\mathbb P_n(\mathbf 1_k\otimes\mathbb P_{n-k})=\mathbb P_n.
$$
\end{proof}

\begin{proof}[Proof of Proposition~\ref{gdrd6i}]
The result below was shown in the proof of Theorem~3.1 in \cite{bozejkospeicherMathAnn1994}.

\begin{lemma}[\cite{bozejkospeicherMathAnn1994}]\label{tyre6i4}
Let a bounded linear operator $\mathcal R_n:\mathcal H^{\otimes n}\to \mathcal H^{\otimes n}$ be defined by
\begin{equation}\label{cy7r5i}
\mathcal R_n:=\mathbf 1_n+\Psi_1+\Psi_1\Psi_2+\dots+\Psi_1\Psi_2\dotsm \Psi_{n-1}.\end{equation}
Then, for $n\in\mathbb N$,
\begin{equation}\label{tye57i}
(n+1)\mathcal P_{n+1}=(\mathbf 1\otimes\mathcal P_n)\mathcal R_{n+1}.
\end{equation}
\end{lemma}

Analogously to $\mathcal F_{\mathrm{fin}}(\mathcal H)$ we define a linear space $\mathbf F_{\mathrm{fin}}(\mathcal H)$ that consists of finite sequence $(f^{(0)},f^{(1)},\dots,f^{(n)},\dots)$ with $f^{(i)}\in\mathcal H^{\otimes i}$. For $h\in\mathcal H$, we define a linear operator $A^-(h):\mathbf F_{\mathrm{fin}}(\mathcal H)\to\mathbf F_{\mathrm{fin}}(\mathcal H) $ by setting
$$ (A^-(h)f^{(n)})(t_1,\dots,t_{n-1}):=\int_T \overline{h(s)}f^{(n)}(s,t_1,\dots,t_{n-1})\,\sigma(ds)$$
for $f^{(n)}\in\mathcal H^{\otimes n}$, $n\in\mathbb N$, and $A^-(h)(1,0,0,\dots):=0$.

Let $f^{(n)}\in\mathcal F_n(\mathcal H)$,  $g^{(n+1)}\in\mathcal H^{\otimes (n+1)}$, and $h\in\mathcal H$. Then, by \eqref{uf7re7ic} and Lemmas~\ref{dryde6644e36} and \ref{tyre6i4}, we get
\begin{align}
&\left(a^+(h)f^{(n)},\mathbb P_{n+1}g^{(n+1)}\right)_{\mathcal F_{n+1}(\mathcal H)}(n+1)!\notag\\
&\quad=\left(\mathcal P_{n+1}\mathbb P_{n+1} (h\otimes f^{(n)}), \mathbb P_{n+1}g^{(n+1)}\right)_{\mathcal H^{\otimes(n+1)}}(n+1)!\notag\\
&\quad=\left(\mathcal P_{n+1} (h\otimes f^{(n)}), g^{(n+1)}\right)_{\mathcal H^{\otimes(n+1)}}(n+1)!\notag\\
&\quad=\left(\mathcal R_{n+1}^*(\mathbf 1\otimes\mathcal P_n) (h\otimes f^{(n)}), g^{(n+1)}\right)_{\mathcal H^{\otimes(n+1)}}n!\notag\\
&\quad =\left(h\otimes (\mathcal P_n f^{(n)}), \mathcal R_{n+1} g^{(n+1)}\right)_{\mathcal H^{\otimes(n+1)}}n!\notag\\
&\quad=\left(\mathcal P_n f^{(n)},A^-(h)  \mathcal R_{n+1} g^{(n+1)}\right)_{\hcm} n!\notag\\
&\quad=\left( f^{(n)},\mathbb P_nA^-(h)  \mathcal R_{n+1} g^{(n+1)}\right)_{\mathcal F_n(\mathcal H)} n!\,.\notag
\end{align}
 From here both formulas \eqref{d6ueu6eu6} and  \eqref{ty456ew} follow.
\end{proof}

\begin{rem}\label{yuftut8}Note that formula \eqref{ty456ew} can now be written in the form
\begin{equation}\label{uy96y9p}
a^-(h)\mathbb P_n g^{(n)}=\mathbb P_{n-1}A^-(h)\mathcal R_n g^{(n)}
\end{equation}
for $h\in\mathcal H$ and $g^{(n)}\in\hcm$.
\end{rem}

\begin{proof}[Proof of Theorem~\ref{Q-CR}]
By choosing an orthonormal basis $(e_n)_{n\in\mathbb N}$ of $\mathcal H$ and writing the infinite matrix of the operator $\Psi$ (see \eqref{tye654i}) in terms of the orthonormal basis $(e_n\otimes e_m)_{n,m\in\mathbb N}$ of $\mathcal H^{\otimes 2}$, one can derive the commutation relation \eqref{byr75ro5} from Section~3 of \cite{bozejkospeicherMathAnn1994}. For the reader's convenience, we will now present a complete  proof of this commutation relation without use of an orthonormal basis.

Let $g,h\in\mathcal H$ and $f^{(n)}\in\mathcal F_n(\mathcal H)$. By formulas \eqref{uf7re7ic} and \eqref{uy96y9p},
 we get
\begin{equation}
a^-(g)a^+(h)f^{(n)}=\mathbb P_n A^-(g)\mathcal R_{n+1}(h\otimes f^{(n)}).\label{vgtyf76ir5ei}
\end{equation}
By \eqref{cy7r5i},
\begin{equation}\label{eq4tqpup}
\mathcal R_{n+1}=\mathbf 1_{n+1}+\Psi_1(\mathbf 1\otimes\mathcal R_{n}).\end{equation}
Formulas \eqref{vgtyf76ir5ei} and  \eqref{eq4tqpup} yield
\begin{equation}\label{tyre6e64sxsx}
a^-(g)a^+(h)f^{(n)}=(g,h)_{\mathcal H}f^{(n)}+
\mathbb P_n u^{(n)},
\end{equation}
where
$$
u^{(n)}:=A^-(g)\Psi_1\big(h\otimes(\mathcal R_n f^{(n)})\big).
$$
A direct calculation shows that
\begin{equation}\label{tur6574e}
u^{(n)}(t_1,\dots,t_n)=\int_T\sigma(ds)\, \overline{g(s)}\,h(t_1)Q(s,t_1)\big(\mathcal R_n f^{(n)}\big)(s,t_2,\dots,t_n).
\end{equation}

On the other hand, using additionally \eqref{gyfg8q7}, we get
\begin{equation}
a^+(h)a^-(g)f^{(n)}=\mathbb P_n\left(h\otimes
(\mathbb P_{n-1}A^-(g)\mathcal R_n f^{(n)})\right)
=\mathbb P_nv^{(n)},\label{tyr6e438w}
\end{equation}
where
\begin{equation*}
v^{(n)}:=h\otimes
(A^-(g)\mathcal R_n f^{(n)}).
\end{equation*}
Note that
\begin{equation}\label{uyrt67r}
v^{(n)}(t_1,\dots,t_n)=\int_T\sigma(ds)\, \overline{g(s)}\,h(t_1)\big(\mathcal R_n f^{(n)}\big)(s,t_2,\dots,t_n).\end{equation}
Formulas \eqref{tyre6e64sxsx}--\eqref{uyrt67r} prove \eqref{byr75ro5}.

Corollary~\ref{uyt67ct} and formula \eqref{rdrt} show that, for each $\varphi^{(2)}\in\mathcal H^{\otimes 2}$ and $f^{(n)}\in\mathcal F_n(\mathcal H)$,
\begin{equation}\label{drter6u}
\int_{T^2}\sigma(ds)\,\sigma(dt)\, \varphi^{(2)}(s,t)\,\partial_s^\dag\partial_t^\dag f^{(n)}=\mathbb P_{n+2}\big((\mathbb P_2\varphi^{(2)})\otimes f^{(n)}\big).
\end{equation}
By Theorem~\ref{main theorem for range of pn}, since $\varphi^{(2)}$ has support in $\Theta$, we get $\mathbb P_2\Psi\varphi^{(2)}=\mathbb P_2\varphi^{(2)}$. Hence, formulas \eqref{hyyt9} and \eqref{drter6u} imply
\begin{align*}
\int_{T^2}\sigma(ds)\,\sigma(dt)\, \varphi^{(2)}(s,t)\,\partial_s^\dag\partial_t^\dag f^{(n)}&=\int_{T^2}\sigma(ds)\,\sigma(dt)\, Q(s,t)\varphi^{(2)}(t,s)\,\partial_s^\dag\partial_t^\dag f^{(n)}\\
&=\int_{T^2}\sigma(ds)\,\sigma(dt)\, \varphi^{(2)}(s,t)Q(t,s)\,\partial_t^\dag\partial_s^\dag f^{(n)},
\end{align*}
 which gives \eqref{ft7r54}.

 Finally, formula \eqref{crte64ue}  is obtained by taking the adjoint operators on the left and right hand sides of formula  \eqref{ft7r54}, see \eqref{jgfruy7r}.
\end{proof}

\begin{proof}[Proof of Theorem \ref{Theorem exclusion principle}]
 Using Corollary \ref{uyt67ct}, we get, for each $f^{(n)}\in\mathcal F_n(\mathcal H)$ and $h\in\mathcal H$,
 \begin{equation*}\label{t7r7r}
 a^+(h)^mf^{(n)}=\mathbb P_{m+n}(h^{\otimes m}\otimes f^{(n)})=\mathbb P_{m+n}((\mathbb P_m h^{\otimes m})\otimes f^{(n)}).\end{equation*}
 Hence, it suffices to prove that $\mathbb P_m(h^{\otimes m})=0$.

Denote by $(e_t)_{t\in T}$ the canonical orthonormal basis in $\mathcal H=\ell^2(T\to\mathbb C)$, i.e.,  $e_t(s)=1$ if $s=t$ and $e_t(s)=0$ if $s\ne t$. In view of \eqref{tye654i}, we get
$$\Psi e_s\otimes e_t=Q(t,s)e_t\otimes e_s,\quad (s,t)\in T^2.$$
Note that the operators $(\Psi_\pi)_{\pi\in S_m}$ form a unitary representation of the group $S_m$, see the proof of Lemma~\ref{tydf6ir5}. Therefore, for each  $k\in\{1,\dots,m-1\}$, we have  $\mathbb P_m\Psi_k=\mathbb P_m$. Hence, for any $t_1,\dots,t_m\in T$,
\begin{equation*}
\mathbb P_m(e_{t_1}\otimes\dots\otimes e_{t_{k-1}}\otimes e_{t_{k+1}}\otimes e_{t_k}\otimes e_{t_{k+2}}\otimes\dots\otimes e_{t_m})
=Q(t_{k},t_{k+1})\mathbb P_n(e_{t_1}\otimes\dots\otimes e_{t_m}).\end{equation*}
This implies that
\begin{equation}\label{discrete fermionic property}
\mathbb P_m(e_{t_1}\otimes \dots \otimes e_{t_{m}})=0 \quad \text{if } \left|\{t_1, \dots , t_m\}\right|<m
\end{equation} (i.e., if some index $t_i$ appears twice or more times).
Analogously, for any $(t_1,\dots,t_m)\in T^m$ and $\pi\in S_m$,
\begin{equation}\label{ioy987}
\mathbb P_m(e_{t_{\pi(1)}}\otimes\dots\otimes e_{t_{\pi(m)}})=Q_\pi(t_1,\dots,t_m)\mathbb P_m(e_{t_1}\otimes \dots \otimes e_{t_{m}}).
\end{equation}

 Let $h=\sum_{t\in T}h_te_t\in\mathcal H$. We get, by \eqref{discrete fermionic property} and \eqref{ioy987},
\begin{align}
\mathbb P_m h^{\otimes m}&=\sum_{t_1,\dots,t_m\in T}h_{t_1}\dotsm h_{t_m}
\mathbb P_m(e_{t_1}\otimes\dots\otimes e_{t_m})\notag\\
&=\sum_{\substack{t_1,\dots,t_m\in T\\
t_i\ne t_j\ \text{if }i\ne j}}h_{t_1}\dotsm h_{t_m}
\mathbb P_m(e_{t_1}\otimes\dots\otimes e_{t_m})\notag\\
&=\sum_{\substack{t_1,\dots,t_m\in T\\
t_1<t_2\dots<t_m}}\sum_{\pi\in S_m}
h_{t_{\pi(1)}}\dotsm h_{t_{\pi(m)}}
\mathbb P_m(e_{t_{\pi(1)}}\otimes\dots\otimes e_{t_{\pi(m)}})\notag\\
&=\sum_{\substack{t_1,\dots,t_m\in T\\
t_1<t_2\dots<t_m}}h_{t_1}\dotsm h_{t_m}\sum_{\pi\in S_m}
\mathbb P_m(e_{t_{\pi(1)}}\otimes\dots\otimes e_{t_{\pi(m)}})\notag\\
&=\sum_{\substack{t_1,\dots,t_m\in T\\
t_1<t_2\dots<t_m}} h_{t_1}\dotsm h_{t_m}
\left(\sum_{\pi\in S_m}
Q_\pi(t_1,\dots,t_m)\right)\mathbb P_m(e_{t_1}\otimes\dots\otimes e_{t_m}),\label{tyei67eii}
\end{align}
where we used that
$h_{t_{\pi(1)}}\dotsm h_{t_{\pi(m)}}=h_{t_1}\dotsm h_{t_m}$
for any permutation $\pi\in S_m$.
 It can be easily proven by induction on $m$ that, for any $t_1,\dots,t_m\in T$ with $t_1<t_2\dots<t_m$, we have
\begin{equation}\label{tye6ue6iess}
\sum_{\pi\in S_m}Q_\pi(t_1,\dots,t_m)=[m]_q!\,.
\end{equation}
Here we used the notation, for $m\in\mathbb N$ and $q\ne1$,
$$
[m]_q!:=\prod_{i=1}^{m}[i]_q,\quad\text{where }[i]_q:=
1+q+q^2+\dots+q^{i-1}=\frac{1-q^i}{1-q}. $$
Since $q^m=1$, we get $[m]_q!=0$. Hence, the theorem follows from \eqref{tyei67eii} and \eqref{tye6ue6iess}.
  \end{proof}

\section*{Acknowledgements}
The authors acknowledge the financial support for this work by the Polish National Science Center, grant no.~2012/05/B/ST1/00626. M.B. was partially supported by the MAESTRO grant DEC-2011/02/A/ST1/00119. M.B. and E.L. acknowledge the financial support of
the SFB 701 `Spectral structures and topological methods in mathematics,' Bielefeld University.

%\section*{Acknowledgements}
%The authors acknowledge the financial support for this work by the Polish National Science Center, grant no.~2012/05/B/ST1/00626. M.B. and E.L. acknowledge the financial support of
%the SFB 701 `Spectral
%structures and topological methods in mathematics,' Bielefeld University.

\end{document}